\documentclass[a4paper,11pt]{article}

\usepackage{fullpage}
\usepackage{times}
\usepackage{soul}
\usepackage{url}
\usepackage[utf8]{inputenc}
\usepackage[small]{caption}
\usepackage{graphicx}
\usepackage{amsmath}
\usepackage{booktabs}

\usepackage[ruled,vlined]{algorithm2e}

\usepackage{enumitem}

\usepackage{libertine}

\usepackage{amssymb,amsmath,amsfonts,amstext,amsthm}

\usepackage{hyperref}
\usepackage[svgnames]{xcolor}
\usepackage[capitalise,nameinlink]{cleveref}
\hypersetup{colorlinks={true},linkcolor={DarkBlue},citecolor=[named]{DarkGreen}}

\usepackage{tikz}  %needed for \textcolor as well as tikz
\usetikzlibrary{arrows}
\usetikzlibrary{patterns,snakes}
\usetikzlibrary{decorations.shapes}
\tikzstyle{overbrace text style}=[font=\tiny, above, pos=.5, yshift=5pt]
\tikzstyle{overbrace style}=[decorate,decoration={brace,raise=5pt,amplitude=3pt}]
\usetikzlibrary{shapes.geometric}
\usetikzlibrary{shapes,positioning}
\usetikzlibrary{calc,positioning}
\usetikzlibrary{shapes.multipart}

\definecolor{cadmiumgreen}{rgb}{0.0, 0.42, 0.24}

\usepackage{bbm}

\newtheorem{theorem}{Theorem}[section]
\newtheorem{corollary}[theorem]{Corollary}

\newtheorem{lemma}[theorem]{Lemma}

\theoremstyle{definition}

\usepackage[numbers]{natbib}

\usepackage{mathtools}
\usepackage{xcolor} %For colors in comments, we can remove afterwards.

\usepackage{authblk}

\usepackage{subcaption}
\newcommand{\calN}{\mathcal{N}}

\newcommand{\calA}{\mathcal{A}}
\newcommand{\calM}{\mathcal{M}}
\newcommand{\calD}{\mathcal{D}}

\newcommand{\calZ}{\mathcal{Z}}
\newcommand{\calI}{\mathcal{I}}
\newcommand{\calJ}{\mathcal{J}}
\newcommand{\dist}{\mathtt{dist}}

\newcommand{\SC}{\text{\normalfont SC}}

\newcommand{\xx}{\mathbf{x}}

\newcommand{\pp}{\mathbf{p}}
\newcommand{\yy}{\mathbf{y}}

\newcommand{\DDM}{\text{\sc DistributedMedian}}
\newcommand{\ddm}{\text{\normalfont DM}}
\newcommand{\DMM}{\text{\sc MinimizeMedian}}
\newcommand{\dmm}{\text{\normalfont MM}}
\newcommand{\CDM}{\text{\sc ContinuousDistributedMedian}}
\newcommand{\cdm}{\text{\normalfont CDM}}

\newcommand{\RR}{\mathbb{R}}
\newcommand{\opt}{\text{\normalfont OPT}}

\newcommand{\wc}{\text{\normalfont wc}}

\begin{document}

\allowdisplaybreaks

\title{\bf Approximate mechanism design for \\ distributed facility location}

\author[1]{Aris Filos-Ratsikas}
\author[2]{Alexandros A. Voudouris}

\affil[1]{Department of Computer Science, University of Liverpool, UK}
\affil[2]{School of Computer Science and Electronic Engineering, University of Essex, UK}

\renewcommand\Authands{ and }
\date{}

\maketitle   

\begin{abstract}	
We consider a single-facility location problem, where agents are positioned on the real line and are partitioned into multiple disjoint districts. The goal is to choose a location (where a public facility is to be built) so as to minimize the total distance of the agents from it. This process is distributed: the positions of the agents in each district are first aggregated into a representative location for the district, and then one of the district representatives is chosen as the facility location. This indirect access to the positions of the agents inevitably leads to inefficiency, which is captured by the notion of {\em distortion}. We study the discrete version of the problem, where the set of alternative locations is finite, as well as the continuous one, where every point of the line is an alternative, and paint an almost complete picture of the distortion landscape of both general and strategyproof distributed mechanisms.

\end{abstract}    

%%%%%%
%%%%%%

\section{Introduction}\label{sec:intro}
Social choice theory deals with the aggregation of different, often contrasting opinions into a common decision. There are many applications where the nature of the aggregation process is distributed, in the sense that it is performed in the following two steps: smaller groups of people first reach a consensus, and then their representative choices are aggregated into a final collective decision.
This can be due to multiple reasons, such as scalability (local decisions are much easier to coordinate when dealing with a large number of individuals), or the inherent roles of the participants (for example, being member states in the European Union or electoral bodies in different regional districts). However, although often necessary, this distributed nature is known to lead to outcomes that do not accurately reflect the views of society. A prominent example of this fact is the 2016 US presidential election, where Donald Trump won despite receiving only 46.1\% of the popular vote, as opposed to Hillary Clinton's 48.2\%.

To quantify the inefficiency that arises in distributed social choice settings, recently
Filos-Ratsikas {\em et al.}~\cite{filos2020distortion} adopted and extended the notion of \emph{distortion},
which is broadly used in social choice theory to measure the deterioration of an aggregate objective (typically the utilitarian social welfare) due to the lack of complete information, and thus provides a systematic way of comparing different mechanisms. 
In their work, Filos-Ratsikas {\em et al.} considered a very general social choice scenario with unrestricted agent preferences, and showed asymptotically tight upper and lower bounds on the distortion of plurality-based mechanisms. We follow a similar approach in this paper for a fundamental structured domain of agent preferences, the well-known \emph{facility location} problem on the line of real numbers. 

The facility location problem is one of the most important in social choice, and has been considered in both the economics and  the computer science literature. It is a special case of the \emph{single-peaked preferences} domain \cite{black1986theory,moulin1980strategy} equipped with linear agent cost functions. 
Furthermore, it is the most prominent setting where the agents have \emph{metric preferences}, and as such it has been studied extensively in the related distortion literature for centralized settings \cite{anshelevich2017randomized,anshelevich2018approximating}. 
Finally, facility location was the paradigm used by Procaccia
and Tennenholtz~\cite{procaccia2009approximate} to put forward their agenda of \emph{approximate mechanism design without money}, which resulted in a plethora of works in computer science ever since. 

In the agenda of Procaccia and Tennenholtz, the goal is to design mechanisms that are \emph{strategyproof} (that is, they do not provide incentives to the agents to lie about their true preferences) and have good performance in terms of some aggregate objective, as measured by having low \emph{approximation ratio}. The need for approximation  now comes from the strategyproofness requirement, rather than the lack of information. In fact, the distortion and the approximation ratio are essentially two sides of the same coin, differentiated by the reason for the loss in efficiency. We will be concerned with \emph{distributed mechanisms}, both strategyproof and not, in a quest to quantify the effect of distributed decision making on facility location, both independently and in conjunction with strategyproofness. Hence, our work follows the agendas of both approximate mechanism design \cite{procaccia2009approximate} and of distributed distortion \cite{filos2020distortion}, and can be cast as  \emph{approximate mechanism design for distributed facility location}.

\subsection{Our setting and contribution}
We study the distributed facility location problem on the real line $\RR$. As in the standard centralized problem, there is a set of agents with ideal positions and a set of alternative locations where the facility can be built. We consider both the \emph{discrete} setting, where the set of alternatives is some finite subset of $\RR$, as well as the \emph{continuous} setting, where the set of alternatives is the whole $\RR$. 
In the distributed version, the agents are partitioned into \emph{districts}, and the aggregation of their positions into a single facility location is performed in two steps: In the first step, the agents of each district select a \emph{representative} location for their district, and in the second step, {\em one of the representatives} is chosen as the final facility location; in Section~\ref{sec:extensions}, we discuss how our results extend to the case of {\em proxy voting}, where the location can be chosen from the set of all alternatives.

Our goal is to find the mechanism with the smallest possible \emph{distortion}, which is defined as the worst-case ratio (over all instances of the problem) between the social cost of the location chosen by the mechanism and the minimum social cost over all locations;  the social cost of a location is the total distance between the agent positions and the location. Note that the optimal location is calculated as if the agents are {\em not} partitioned into districts, and thus the distortion accurately measures the effect of selecting the facility location in a distributed manner to the efficiency of the system. 
We are also interested in \emph{strategyproof mechanisms}, for which the distortion quantifies the loss in performance both due to lack of information and due to requiring strategyproofness.
We mainly focus on the case of {\em symmetric} districts, which have equal size; in Section~\ref{sec:extensions} we discuss the case of {\em asymmetric} districts and other extensions. 
Our results are as follows (see also \cref{tab:results}): 
\begin{itemize}
\item For the discrete setting, 
the best possible distortion by any mechanism is $3$, 
and the best possible distortion by any strategyproof mechanism is $7$.

\item For the continuous setting, 
the best possible distortion by any mechanism is between $2$ and $3$, 
and the best possible distortion by any strategyproof mechanism is $3$.
\end{itemize}

\setlength{\tabcolsep}{1em}
\begin{table}
\centering
{\renewcommand{\arraystretch}{1.3}{
\begin{tabular}{lcc}
\noalign{\hrule height 1pt}\hline
 						    & Discrete 		& Continuous	    \\
\noalign{\hrule height 1pt}\hline
General mechanisms   	    & $3$                & $\in [2,3]$     \\ 
Strategyproof mechanisms    & $7$                & $3$                 \\ 
\noalign{\hrule height 1pt}\hline
\end{tabular}
}}
\caption{An overview of our bounds on the distortion of general and strategyproof mechanisms for discrete and continuous distributed facility location, when the districts are symmetric. 
In the discrete setting, the lower bound of $7$ also holds for {\em ordinal} (not necessarily strategyproof) mechanisms. The upper bound of $3$ for general mechanisms is achieved by $\DMM$, while the upper bound of $7$ for strategyproof mechanisms is achieved by $\DDM$, which is an ordinal mechanism. 
In the continuous setting, the upper bound of $3$ for general and strategyproof mechanisms is achieved by the continuous version of $\DDM$, which is actually an implementation of $\DMM$.}
\label{tab:results}
\end{table}

The mechanisms we design are adaptations of well-known mechanisms for the centralized facility location problem. In the discrete setting, the mechanism with the best possible distortion of $3$ selects the representative of each district to be the location that minimizes the social cost of the agents therein, and then chooses the median representative as the facility location; we refer to this mechanism as $\DMM$. By modifying the first step so as to select the representative of a district to be the location that is the closest to the median agent in the district, we obtain the $\DDM$ mechanism, which is the best possible strategyproof mechanism with distortion $7$.  When we move to the continuous setting, selecting the median agent within each district minimizes the social cost of the agents therein, and thus $\DDM$ is an implementation of $\DMM$. The proofs of our upper bounds in Sections~\ref{sec:discrete} and~\ref{sec:continuous} rely on a characterization of the structure of worst-case instances (in terms of distortion) for each of these mechanisms, which is obtained by carefully modifying the positions of some agents without decreasing the distortion.

For the lower bounds, we employ the following main idea. We construct instances of the problem for which any mechanism with low distortion (depending on the bound we are aiming for) must satisfy some constraints about the representative $y$ it can choose for a particular district, namely, either that $y$ is some specific location (in the discrete setting), or that it must lie in some specific interval (in the continuous setting). Then, because of the distributed nature of the mechanism, we can exploit the fact that $y$ must represent this district in any instance that contains it, and use such instances to either argue about the distortion of the mechanism, or to impose constraints on the representatives of other districts. At the heart of all of our constructions lies one type of crucial lemma (see \cref{lem:011winner-general} for an example) which establishes that, at least for the instances we consider, any mechanism must select the median representative in the second step of aggregation. The proofs of these lemmas are similar in the sense that they use the idea highlighted above repeatedly and inductively, and in conjunction with arguments involving strategyproofness when necessary. However, they are also notably different because they apply to different settings (discrete vs continuous) or to mechanisms with different distortion bounds ($3$ or $2$ for general vs $7$ or $3$ for strategyproof).

Interestingly, when given as input the particular instances we use in the proof of the lower bound in the discrete setting, strategyproof mechanisms exhibit an {\em ordinal} behavior. Consequently, the very same proof can be used to show that the lower bound of $7$ also holds for {\em ordinal mechanisms}, which do not take into account the actual positions of the agents, but instead base their decisions only on the preference rankings that the positions of the agents induce over the alternative locations. Furthermore, this bound is tight since $\DDM$ is in fact ordinal (whereas $\DMM$ is not). Finally, observe that ordinality is not a meaningful property in the continuous setting, as every single position induces a different preference ranking over locations.

\subsection{Related work}\label{sec:related}
The notion of the distortion of social choice mechanisms, as well as the corresponding research agenda, was initiated by \citet{procaccia2006distortion}, who considered an unrestricted preference setting in which the agents have {\em normalized} cardinal valuations and the objective is to choose a single winning alternative.
A subsequent stream of papers studied several variants of the problem, 
including 
the original single-winner setting~\citep{boutilier2015optimal,Caragiannis2011embedding}, 
multi-winner elections~\citep{caragiannis2017subset}, 
participatory budgeting~\citep{benade2017preference,goel2016knapsack}, 
as well as settings showcasing tradeoffs between the distortion and cardinal information~\citep{ABFV20,mandalefficient,mandalec20},
Moreover, there are quite a few papers that have studied the distortion of strategyproof mechanisms~\citep{bhaskar2018truthful,caragiannis2018truthful,filos2014truthful,Aris14}. 
In its original definition, the distortion measured the performance of ordinal mechanisms in terms of a cardinal objective, namely the utilitarian social welfare (the total utility of the agents for the chosen outcome). However, if one interprets more generally the lack of information as the reason for the loss in efficiency, the distortion actually captures much wider scenarios, like the distributed social choice setting studied by \citet{filos2020distortion}.

Although the number of papers dealing with (variants of) the aforementioned normalized setting is substantial, the literature on the distortion flourished after the work of \citet{anshelevich2018approximating} and \citet{anshelevich2017randomized}, who studied settings in which the agents have {\em metric} preferences. 
Such preferences are constrained by the fact that the utility (or cost in the particular case) of every agent for different alternatives must satisfy the triangle inequality, which effectively results in the distortion bounds being small constants, rather than asymptotic bounds depending on the number of agents and alternatives, as it is typically the case in the normalized setting. 
Similar investigations have given rise to a plethora of papers on this topic; see~\citep{abramowitz2017utilitarians,abramowitz2019awareness,anshelevich2016blind,anshelevich2017tradeoffs,anshelevich2018ordinal,fain2018random,feldman2016facility,goel2017metric,goel2018relating,gross2017agree,kempe2019analysis,munagala2019improved}. 
For a comprehensive introduction to the distortion literature, we refer the reader to the survey of Anshelevich {\em et al.}~\citep{survey}.

As already mentioned earlier, the facility location problem plays an important role in the literature at the intersection of computer science and economics. From a purely algorithmic perspective, facility location problems have a long history in the area of approximation~(e.g., see \citep{shmoys1997approximation}). At the same time, many works in economics have studied such problems \citep{bogomolnaia2007euclidean,peters1992pareto,rader1963existence} in the context of \emph{Euclidean preferences} \citep{hotelling1990stability}, a special case of the celebrated class of \emph{single-peaked preferences} \citep{black1986theory,moulin1980strategy}; see also \citep{elkind2014recognizing,peters2017recognising}. 
The problem became extremely popular in the economics and computation community after \citet{procaccia2009approximate} used it to put forward their agenda of \emph{approximate mechanism design without money}, following the similar agenda of \citet{nisan2001algorithmic} for settings with money.  Since then, the facility location problem has been studied extensively, for 
different objectives~\citep{alon2009strategyproof,cai2016facility,feigenbaum2013approximately,feldman2013strategyproof}, 
multiple facilities~\citep{escoffier2011strategy,fotakis2013power,lu2009tighter,lu2010asymptotically}, 
different domains~\citep{schummer2002strategy,tang2020characterization,sui2013analysis,sui2015approximately}, 
different cost functions \cite{filos2015facility,fotakis2013strategyproof}, and 
several variants of the problem~\cite{cheng2013obnoxious,cheng2011mechanisms,DFV21,duan2019heterogeneous,fong2018facility,fotakis2013winner,serafino2015truthful,serafino2016heterogeneous}; See also the recent survey of Chan {\em et al.}~\citep{FLsurvey}.

The most related setting to our work is an extension studied by Procaccia and Tennenholtz~\cite{procaccia2009approximate} with agents (or, super-agents, for clarity) controlling multiple locations, whose cost is the total distance between their locations and the facility. They showed that the mechanism that first selects the median location of each super-agent and then the median of those is strategyproof and $3$-approximate for the social cost. This implies an upper bound of $3$ on the distortion of strategyproof mechanisms in our continuous setting, by interpreting the super-agents as district representatives; we show that this bound can be obtained by simple extensions of our techniques for the discrete setting. Procaccia and Tennenholtz also showed a matching lower bound, which however requires the super-agents to be truthful, and thus does not have any implications for our setting. This model was later extended by Babaioff {\em et al.}~\cite{babaioff2016mechanism} to a setting where the locations are themselves strategic agents, and the agents of the higher level are strategic \emph{mediators}.

%%%%%%
%%%%%%

\section{Preliminaries}\label{sec:prelim}

We consider the {\em discrete} and the {\em continuous} distributed facility location problem. In both settings, there is a set $\calN$ of $n$ {\em agents} who are positioned on the line of real numbers; let $x_i \in \RR$ denote the {\em position} of agent $i$, and denote by $\xx=(x_i)_{i \in \calN}$ the {\em position profile} of all agents. The agents are partitioned into $k \geq 2$ {\em districts}; let $\calD$ be the set of districts, and denote by $d(i)$ the district containing agent $i$. Let $\calN_d = \{i \in \calN: d(i) = d\}$ be the set of agents that belong to district $d \in \calD$. 
We consider {\em symmetric} districts, which consist of the same number of agents $\lambda = \frac{n}{k}$.  
We will use the notation $\xx_d = (x_i)_{i \in \calN_d}$ for the restriction of $\xx$ to the positions of the agents in district $d$, 
and we will refer to $\xx_d$ as a \emph{district position profile}.
We say that two districts $d$ and $d'$ are \emph{identical} if $\xx_d = \xx_{d'}$.

For two points $x, y \in \RR$, let $\delta(x,y) =\delta(y,x) = |x-y|$ denote their absolute distance. Given a position profile $\xx$, the {\em social cost} of a point $z \in \RR$ is the total distance of the agents from $z$:
$$\SC(z| \xx) = \sum_{i \in \calN} \delta(x_i,z).$$
Our goal is to select a location $z^*$ from a set of {\em alternative locations} $\calZ \subseteq \RR$ to minimize the social cost, that is,
$$z^* \in \arg\min_{z \in \calZ} \SC(z | \xx).$$
In the discrete setting, the set of alternative locations is finite and denoted by $\calA$.
On the other hand, in the continuous setting, the set of alternative locations is the whole $\RR$. 
Therefore, we have that either $\calZ = \calA$ in the discrete version, or $\calZ=\RR$ in the continuous version. 

We use the term \emph{instance} to refer to a tuple $\calI=(\xx,\calD,\calZ)$ consisting of a position profile $\xx$, a set of districts $\calD$, and a set of alternative locations $\calZ$; we omit the set of agents $\calN$ from the definition of the instance as it is implied by $\xx$. In the continuous setting, since the set of alternative locations is clear, we will simplify our notation and use a pair $(\xx,\calD)$ to denote an instance. 

If we had access to the positions of all the agents, it would be easy to select the optimal location in both versions of the problem. However, in our setting this is not possible as the positions are assumed to not be globally known, only \emph{locally}. To decide the facility location we deploy {\em distributed mechanisms}, to which we will simply refer as \emph{mechanisms} from now on.
A mechanism $\calM$ consists of the following two steps of aggregation: 
\begin{itemize}
\item First step: For every district $d \in \calD$, $\calM$ aggregates the positions of the agents therein into the {\em representative} location $z_d \in \calZ$ of $d$. This step is \emph{local}, in the sense that the representative $z_d$ is a result of the corresponding district profile $\xx_d$ only. Formally, for any two instances that contain two identical districts $d_1$ and $d_2$, $\calM$ must choose the same representative for both districts, that is, $z_{d_1}=z_{d_2} \in \calZ$. Essentially, this property stipulates that the representative of a district is chosen only by the members of the district, and independently of agents in other districts.  

\item Second step: $\calM$ aggregates the district representatives into a single {\em facility location}. For a given instance $\calI = (\xx, \calD, \calZ)$, we denote by $\calM(\calI)$ the facility location chosen by the mechanism. 
\end{itemize}

\subsection{The distortion of mechanisms}

Due to the lack of global information, the facility location chosen by a mechanism will inevitably be suboptimal. To quantify this inefficiency, we adopt and extend the notion of distortion to our setting. 
The {\em distortion} of an instance $\calI = (\xx,\calD,\calZ)$ subject to using a mechanism $\calM$ is the ratio between the social cost of the location $\calM(\calI)$ chosen by the mechanism given $\calI$ as input and the social cost of the optimal location $\opt(\calI) = \arg\min_{z \in \calZ} \SC(z|\xx)$ for the instance:
\begin{align*}
\dist(\calI | \calM) = \frac{\SC(\calM(\calI)|\xx)}{ \SC( \opt(\calI) |\xx)}.
\end{align*}
Then, the distortion of mechanism $\calM$  is the worst-case distortion over all possible instances:
\begin{align*}
\dist(\calM) = \sup_{\calI} \dist(\calI | \calM).
\end{align*}

We will now argue that it is without loss of generality to focus on mechanisms satisfying a simple unanimity property (within the districts). In particular, we say that a mechanism $\calM$ is \emph{unanimous}, if it chooses the representative of a district to be $z \in \calZ$,  whenever all agents of the district are positioned at $z$. The proof of the following lemma can be found in the appendix.

\begin{lemma}\label{lem:unanimous}
Any mechanism with finite distortion must be unanimous. 
\end{lemma}	

\subsection{Strategyproofness}
Besides achieving low distortion, we are also interested in mechanisms which ensure that the agents report their positions {\em truthfully}, that is, they have no incentive to misreport hoping to change the outcome of the mechanism to a location that is closer to their position. Formally, let $\calI=(\xx,\calD,\calZ)$ be an instance, where $\xx$ is the {\em true} position profile of the agents, and let $\calJ = (\yy,\calD,\calZ)$ be any instance with position profile $\yy=(y_i,\xx_{-i})$, in which agent $i$ reports $y_i$ and all other agents report their positions according to $\xx$. A mechanism $\calM$ is {\em strategyproof} if the location chosen by $\calM$ when given as input $\calI$ is closer to the position $x_i$ of any agent $i$ than the location chosen by $\calM$ when given as input $\calJ$. In other words, for every agent $i$ and $y_i \in \RR$, it must hold that
\[\delta(x_i,\calM(\xx,\calD,\calZ)) \leq \delta(x_i,\calM( (y_i,\xx_{-i}), \calD,\calZ)).\]  
This added requirement of strategyproofness imposes further restrictions, and potentially impacts the achievable distortion as well. So, our goal is to design strategyproof mechanisms with as low distortion as possible.

We now define the class of mechanisms that are strategyproof within districts. Intuitively, such mechanisms prevent the agents from misreporting in hopes of changing the representative of their district to a location closer to them. Observe that a strategyproof mechanism could in principle allow such a local manipulation, only to eliminate it in the second step (for example, by completely ignoring the representatives and choosing an arbitrary fixed facility location). We show that for mechanisms with a finite distortion, this is impossible.

Formally, a mechanism $\calM$ is \emph{strategyproof within districts} if for any district $d \in \calD$, the representative of $d$ on input $\calI = (\xx,\calD,\calZ)$ is closer to the true position $x_i$ of every agent $i$ than the representative of $d$ on input $\calJ = ((y_i,\xx_{-i}),\calD,\calZ)$. 
We can now show the following useful property of stratefyproof mechanisms; the proof is deferred to the appendix.

\begin{lemma}\label{lem:spiout}
Any strategyproof mechanism with finite distortion is strategyproof within districts.	
\end{lemma}

%%%%%%
%%%%%%

\section{Mechanisms for the discrete setting}\label{sec:discrete}

We begin the exposition of our results from the discrete setting. We consider two natural mechanisms, which we call $\DMM$ ($\dmm$) and $\DDM$ ($\ddm$). Given the representatives of the districts, both mechanisms select the facility location to be the median representative. The main difference between the two mechanisms is on how they select the representatives of the districts: $\dmm$ selects the representative of each district to be the alternative location that minimizes the social cost of the agents within the district, while $\ddm$ 
selects the representative of each district to be the location which is closer to the median agent in the district.
In case there are at least two median representatives or at least two locations minimizing the social cost within some district, the mechanisms select the left-most such option.

As one might expect, the fact that $\dmm$ minimizes the social cost within the districts may give the opportunity to some agents therein to misreport their positions hoping to affect the outcome. On the other hand, by choosing the median location both within and over the districts, $\ddm$ does not allow such manipulations. Formally, we have the following statement, whose proof is deferred to the appendix.

\begin{theorem}\label{thm:sp-mechanisms}
$\dmm$ is not strategyproof, while $\ddm$ is strategyproof.
\end{theorem} 

In the rest of this section, we focus on bounding the distortion of the two mechanisms. 
To do so, we first show in \cref{sec:discrete-worst-case} that the instances achieving the worst-case distortion have a very particular structure, which is common for both mechanisms. We then exploit this structure in \cref{sec:discrete-distortion} to show an upper bound of $3$ on the distortion of $\dmm$ and an upper bound of $7$ on the distortion of $\ddm$.

\subsection{Worst-case instances}\label{sec:discrete-worst-case}
We start by characterizing the structure of worst-case instances for any mechanism $\calM \in \{\dmm,\ddm\}$. 
Let $\wc(\calM)$ be the class of instances $\calI=(\xx,\calD,\calA)$ such that
\begin{itemize}
\item[(P1)]
For every agent $i \in \calN$, 
\begin{itemize}
\item $x_i \geq \calM(\calI)$ if $\calM(\calI) < \opt(\calI)$, or 
\item $x_i \leq \calM(\calI)$ if $\calM(\calI) > \opt(\calI)$.
\end{itemize}

\item[(P2)]
For every $z \in \calA$ which is representative for a set of districts $\calD_z \neq \varnothing$,  the positions of all agents in the districts of $\calD_z$ are in the interval defined by $z$ and $\opt(\calI)$.
\end{itemize}

\begin{algorithm}[t]
	\SetKwProg{mechanism}{Mechanism}{}{}
	\mechanism{$\dmm(\xx,\calD,\calA)$}{
		\For{each district $d \in \calD$}{
			$z_d \leftarrow \text{left-most location in } \arg\min_{z \in \calA} \sum_{i \in \calN_d} \delta(x_i,z)$
		}
		\Return {\sc Median}$(\{z_d\}_{d \in \calD})$
	}
	
	\vspace{10pt}
	
	\SetKwProg{mechanism}{Mechanism}{}{}
	\mechanism{$\ddm(\xx,\calD,\calA)$}{
		\For{each district $d \in \calD$}{
			$z_d \leftarrow \arg\min_{z \in \calA} \delta(\text{\sc Median}(\xx_d),z)$
		}
		\Return {\sc Median}$(\{z_d\}_{d \in \calD})$
	}
	
	\vspace{10pt}
	
	\SetKwProg{rule}{Rule}{}{}	
	\rule{{\sc Median}$(\yy)$}{
		$\eta \leftarrow |\yy|$ \\
		sort $\yy = (y_1, ..., y_\eta)$ in non-decreasing order \\
		\Return $\yy_{\lfloor \eta/2 \rfloor}$
	}
	
	\caption{The $\DMM$ and $\DDM$ mechanisms.} \label{alg:DMM} 
	
\end{algorithm}

\noindent We will show the following lemma.

\begin{lemma}\label{lem:discrete-structure}
The distortion of $\calM \in \{\dmm,\ddm\}$ is equal to 
$$\sup_{\calI \in \wc(\calM)} \dist(\calI | \calM).$$
\end{lemma}

\begin{proof}
Let $\calM \in \{\dmm,\ddm\}$.
It suffices to show that for every instance $\calJ \not\in \wc(\calM)$, there exists an instance $\calI \in \wc(\calM)$, such that $\dist(\calJ | \calM) \leq \dist(\calI | \calM)$. Due to symmetry, we only focus on the case where $\calM(\calJ) = w < o = \opt(\calJ)$. 
We gradually transform $\calJ$ into $\calI$ as follows:
\begin{itemize}
\item[(T1)] We move every agent with position strictly smaller than $w$ to $w$.
\item[(T2)] For every location $z$ which is representative for a set of districts $\calD_z \neq \varnothing$ in $\calJ$, we move every agent in $\calD_z$ whose position does not lie in the interval defined by $z$ and $o$ to the boundaries of this interval:
\begin{itemize}
\item For $z < o$, if the position of the agent is strictly smaller than $z$ we move her to $z$, and if it is strictly larger than $o$ we move her to $o$;
\item For $z > o$, if the position of the agent is strictly larger than $z$ we move her to $z$, and if it is strictly smaller than $o$ we move her to $o$.
\end{itemize}
\end{itemize}
Observe that, because (T1) is performed before (T2), an agent with position strictly smaller than $w < z < o$ who belongs to a district in $\calD_z$ can be moved twice: once from her initial position to $w$ and then again to $z$; see \cref{fig:upper-discrete} for an example. 
Naturally, these transformations define a sequence of intermediate instances with the same set of districts and alternative locations, but different position profiles. We will show that these instances preserve the following three properties, which are sufficient to show by induction that the distortion does not become smaller as we go from $\calJ$ to $\calI$:
\begin{itemize}
\item The facility location chosen by the mechanism is always $w$;
\item The optimal location is always $o$;
\item For any two consecutive intermediate instances with position profiles $\xx$ and $\yy$, $\frac{\SC(w|\xx)}{\SC(o|\xx)} \leq \frac{\SC(w|\yy)}{\SC(o|\yy)}$.
\end{itemize}

\begin{figure}
	\centering
	\begin{subfigure}[t]{0.45\linewidth}
		\centering
		\begin{tikzpicture}[thick,scale=0.8, every node/.style={scale=0.8}]
		\node [draw,fill=red,inner sep=1.5pt, label={below:$x_i$}] (i) at (-3, 0) {};
		\node [circle, draw=black,color=black,fill=white!50!gray, label={below:$w$}] (w) at (-1.25, 0) {};
		\node [circle, draw=black,color=black,fill=white!50!gray, label={below:$z$}] (z) at (0.75, 0) {};	
		\node [draw,fill=cadmiumgreen,inner sep=1.5pt, label={below:$x_j$}] (j) at (2, 0) {};
		\node [circle, draw=black,color=black,fill=white!50!gray, label={below:$o$}] (o) at (3, 0) {};
		\node [draw,fill=blue,inner sep=1.5pt, label={below:$x_t$}] (t) at (4.75, 0) {};
		\draw (i) to (w);
		\draw (w) to (z);
		\draw (z) to (j);
		\draw (j) to (o);
		\draw (o) to (t);
		\end{tikzpicture}		
		\caption{Initial instance: agent $i$ does not satisfy (P1) and (P2); agent $j$ satisfies both properties; agent $t$ does not satisfy (P2).}
	\end{subfigure}	
	\ \ \ \ \ \ \ \ \ \ \ 
	\begin{subfigure}[t]{0.45\linewidth}
		\centering
		\begin{tikzpicture}[thick,scale=0.8, every node/.style={scale=0.8}]
		\node [draw,fill=gray,inner sep=1.5pt, label={below:$x_i$}] (init) at (-3, 0) {};
		\node [circle, draw=black,color=black,fill=white!50!gray, label={below:$w$}] (w) at (-1.25, 0) {};
		\node [draw,fill=red,inner sep=1.5pt] (i) at (-1.25, 0) {};
		\node [circle, draw=black,color=black,fill=white!50!gray, label={below:$z$}] (z) at (0.75, 0) {};	
		\node [draw,fill=cadmiumgreen,inner sep=1.5pt, label={below:$x_j$}] (j) at (2, 0) {};
		\node [circle, draw=black,color=black,fill=white!50!gray, label={below:$o$}] (o) at (3, 0) {};
		\node [draw,fill=blue,inner sep=1.5pt, label={below:$x_t$}] (t) at (4.75, 0) {};
		\draw (init) to (w);
		\draw (w) to (z);
		\draw (z) to (j);
		\draw (j) to (o);
		\draw (o) to (t);
		\path[draw,-latex]  (init) edge[bend left] (w);
		\end{tikzpicture}		
		\caption{Application of (T1): agent $i$ is moved to $w$ so that (P1) is satisfied.}
	\end{subfigure}
	
	\vspace{15pt}
	
	\begin{subfigure}[t]{0.45\linewidth}
		\centering
		\begin{tikzpicture}[thick,scale=0.8, every node/.style={scale=0.8}]
		\node [draw,fill=gray,inner sep=1.5pt, label={below:$x_i$}] (init) at (-3, 0) {};
		\node [circle, draw=black,color=black,fill=white!50!gray, label={below:$w$}] (w) at (-1.25, 0) {};
		\node [circle, draw=black,color=black,fill=white!50!gray, label={below:$z$}] (z) at (0.75, 0) {};	
		\node [draw,fill=red,inner sep=1.5pt] (i) at (0.75, 0) {};
		\node [draw,fill=cadmiumgreen,inner sep=1.5pt, label={below:$x_j$}] (j) at (2, 0) {};
		\node [circle, draw=black,color=black,fill=white!50!gray, label={below:$o$}] (o) at (3, 0) {};
		\node [draw,fill=blue,inner sep=1.5pt, label={below:$x_t$}] (t) at (4.75, 0) {};
		\draw (init) to (w);
		\draw (w) to (z);
		\draw (z) to (j);
		\draw (j) to (o);
		\draw (o) to (t);
		\path[draw,-latex]  (w) edge[bend left] (z);
		\end{tikzpicture}		
		\caption{Application of (T2): agent $i$ is further moved to $z$ so that (P2) is satisfied.}
	\end{subfigure}
	\ \ \ \ \ \ \ \ \ \ \ 
	\begin{subfigure}[t]{0.45\linewidth}
		\centering
		\begin{tikzpicture}[thick,scale=0.8, every node/.style={scale=0.8}]
		\node [draw,fill=gray,inner sep=1.5pt, label={below:$x_i$}] (init) at (-3, 0) {};
		\node [circle, draw=black,color=black,fill=white!50!gray, label={below:$w$}] (w) at (-1.25, 0) {};
		\node [circle, draw=black,color=black,fill=white!50!gray, label={below:$z$}] (z) at (0.75, 0) {};	
		\node [draw,fill=red,inner sep=1.5pt] (i) at (0.75, 0) {};
		\node [draw,fill=cadmiumgreen,inner sep=1.5pt, label={below:$x_j$}] (j) at (2, 0) {};
		\node [circle, draw=black,color=black,fill=white!50!gray, label={below:$o$}] (o) at (3, 0) {};
		\node [draw,fill=blue,inner sep=1.5pt] (t) at (3, 0) {};
		\node [draw,fill=gray,inner sep=1.5pt, label={below:$x_t$}] (initt) at (4.75, 0) {};
		\draw (init) to (w);
		\draw (w) to (z);
		\draw (z) to (j);
		\draw (j) to (o);
		\draw (o) to (initt);
		\path[draw,-latex]  (initt) edge[bend right] (o);
		\end{tikzpicture}		
		\caption{Application of (T2): agent $t$ is moved to $o$ so that (P2) is satisfied.}
	\end{subfigure}

	\caption{An execution of the transformations used in the proof of \cref{lem:discrete-structure}. The three agents $i$, $j$ and $t$ belong to the same district with representative $z$; $w$ is the facility location chosen by mechanism $\calM \in \{\dmm,\ddm\}$ and $o$ is the optimal location. Transformation (T1) will first move agent $i$ from $x_i$ to $w$ so that (P1) is satisfied, and then transformation (T2) will move both $i$ and $t$ to the boundaries of the interval $[z,o]$ so that (P2) is satisfied. In the proof of the lemma, we show that moving the agents in this way does not affect the facility location chosen by the mechanism nor the optimal location, while at the same time the distortion can only become larger. 
	}
	\label{fig:upper-discrete}
\end{figure}
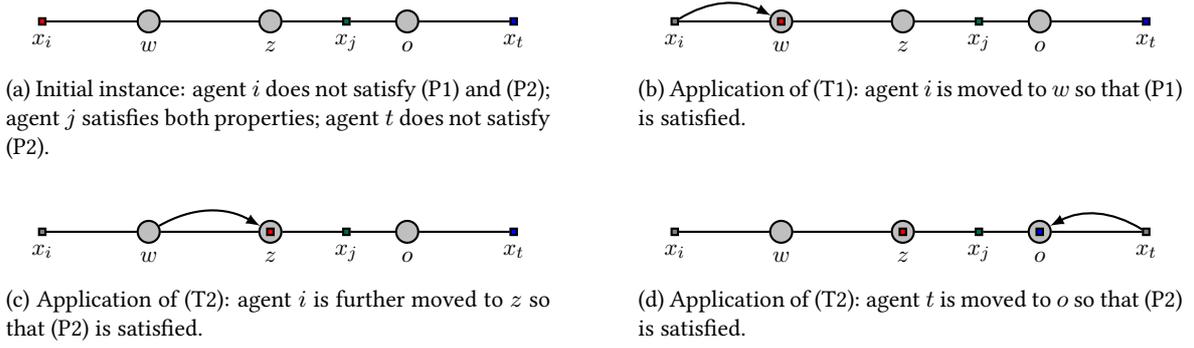

\noindent
Before we continue with the proof of the properties, we state here two useful technical lemmas, whose proofs can be found in the appendix.

\begin{lemma}\label{lem:sc-peaked}
Let $\xx$ be a vector consisting of the positions of a set of agents $S$ such that $\sum_{i \in S} \delta(x_i,z) \leq \sum_{i \in S} \delta( x_i,y)$ for $z < y$. Then, for every $p \in (z,y)$, it holds that $\sum_{i \in S} \delta(x_i,p) \leq \sum_{i \in S} \delta(x_i,y)$. 
\end{lemma}

\begin{lemma}\label{lem:optimality-preservation}
Let $\xx$ be a (district) position profile, and denote by $z$ the optimal location for $\xx$. Then, moving any single agent $i$ with $x_i < p \leq z$ or $x_i > p \geq z$ to $p$, induces a position profile $\pp = (p,\xx_{-i})$ for which the optimal location is again $z$.
\end{lemma}

\bigskip
\noindent
\underline{The facility location is always $w$: $\calM = \dmm$.} \\
For (T1), consider any intermediate instance with position profile $\xx$ such that there exists a district $d \in \calD$ with representative $z_d=z$, which contains some agent $i$ with $x_i < w$ who is moved to $w$. 
We distinguish between two cases:
\begin{itemize}
\item $z > w$. Since $z$ is the optimal location for $\xx_d$, \cref{lem:optimality-preservation} with $p=w$ implies that moving $i$ from $x_i < w < z$ to $w$ does not affect the optimality of $z$. Hence, $z$ remains the representative of $d$, and consequently $w$ remains the facility location chosen by $\dmm$.

\item $z \leq w$. In this case, moving $i$ to $w$, does not necessarily imply that $z$ remains the representative of $d$. However, we claim that the new representative location $y$ can only be such that $y \leq w$, which guarantees that $w$ remains the median representative, and thus the facility location chosen by $\dmm$.

Assume otherwise that the new representative location is $y > w$, in which case the total distance of the agents in $d$ from $y$ is strictly smaller than from $w$ (note that if it was equal, $w$ would become the representative because of the tie-breaking used by the mechanism):
\begin{align*}
\delta(w,y) + \sum_{j \in \calN_d\setminus\{i\}} \delta(x_j,y) < \delta(w,w) + \sum_{j \in \calN_d\setminus\{i\}} \delta(x_j,w).
\end{align*}
However, since $\delta(w,y) = \delta(x_i,y) - \delta(x_i,w)$ and $\delta(w,w)=0$, we equivalently have that
\begin{align*}
\sum_{j \in \calN_d} \delta(x_j,y) < \sum_{j \in \calN_d} \delta(x_j,w),
\end{align*}
which, since $z < w < y$ and $z$ minimizes the total distance under $\xx$, contradicts the fact that the total distance from $w$ is less than the total distance from $y$; this is implied by \cref{lem:sc-peaked} when restricted to the agents in $d$.
\end{itemize}

\noindent
For (T2), consider any intermediate instance with position profile $\xx$ such that there exists an alternative location $z\leq o$ (the case $z < o$ can be handled similarly) which is representative for a set of districts $\calD_z \neq \varnothing$, and some district $d \in \calD_z$ contains an agent $i$ with position $x_i \not\in [z,o]$. Since $z$ is optimal for $\xx_d$, and $i$ is either moved to $z$ if $x_i < z$ or to $o$ if $o < x_i$, \cref{lem:optimality-preservation} (constrained to $\xx_d$, with either $p=z$ or $p=o$ for the two cases, respectively) implies that the optimality of $z$ in $d$ is not affected by moving $i$. Hence, $z$ remains the representative of $d$, and $w$ is still chosen by $\dmm$.

\bigskip
\noindent
\underline{The facility location is always $w$: $\calM = \ddm$.} \\
For (T1), like in the case of $\dmm$, consider any intermediate instance with position profile $\xx$ such that there exists a district $d \in \calD$ with representative $z_d=z$, which contains some agent $i$ with position $x_i < w$ who is moved to $w$. 
We distinguish between two cases:
\begin{itemize}
\item $z > w$. Since $i$ is closer to $w$ than to $z$, she cannot be the median agent in district $d$. 
Therefore, moving agent $i$ to $w$ will not change the representative of $d$, and thus neither the location chosen by $\ddm$.

\item $z \leq w$. By moving agent $i$ to $w$, the representative of $d$ may change from $z$ to $w$, but the location chosen by $\ddm$ will remain the same. In particular, since $w$ is the median among all the district representatives, it will also remain the median after being the representative of more districts.  
\end{itemize}

\noindent
For (T2), consider again any intermediate instance such that there exists an alternative location $z\leq o$ which is representative for a set of districts $\calD_z \neq \varnothing$, and some district of $\calD_z$ contains an agent $i$ with position $x_i \not\in [z,o]$. We distinguish between two cases:
\begin{itemize}
\item $x_i < z$. If $i$ is the median agent in $d$, then the fact that $z$ is the representative of $d$ means that $z$ is the closest location to $i$, and thus moving $i$ to $z$ does change the representative of $d$. On the other hand, if $i$ is not the median, then moving $i$ to $z$ either does not alter who the median agent is (and thus the representative) or $i$ becomes the median and $z$ continues to be the representative. 

\item $x _i> o$. Since $i$ is closer to $o$ than to $z$, the only case in which $i$ might be the median agent of $d$ is when $z=o$, and moving $i$ to $o$ clearly does not affect the representative of $d$. In any other case, $i$ is not the median agent and cannot become the median (as there must exist an agent who is closer to $z$ than to $o$).
\end{itemize}
Therefore, since moving $i$ to the boundaries of $[z,o]$ does not affect the representative of $d$, the location $w$ chosen by $\ddm$ remains the same as well.

\bigskip
\noindent
\underline{The optimal location is always $o$.} \\
Observe that the transformations we perform define two symmetric types of moves: an agent $i$ can be moved either from a position $x_i < p \leq o$ to $p$, or from a position $x_i > p \geq o$ to $p$, where $p \in \calA$. Consequently, \cref{lem:optimality-preservation} immediately implies that any such move does not affect the optimality of $o$ in the induced instance.

\bigskip
\noindent
\underline{The distortion does not get smaller between consecutive instances.} \\
For (T1), consider two consecutive instances with position profiles $\xx$ and $\yy = (w,\xx_{-i})$ in which a single agent $i$ moves from a position $x_i < w$ to $w$. Since $\delta(x_i,o) = \delta(x_i,w) + \delta(w,o)$ and $\delta(w,w)=0$, we have
\begin{align*}
\frac{\SC(w | \xx)}{\SC(o | \xx)} 
&= \frac{\delta(y_i, w) + \sum_{j \neq i} \delta(x_j, w)}{\delta(x_i, o) + \sum_{j \neq i} \delta(x_j, o)} \\
&= \frac{\delta(x_i, w) + \delta(w,w) + \sum_{j \neq i} \delta(x_j, w)}{\delta(x_i, w) + \delta(w,o) + \sum_{j \neq i} \delta(x_i, o)} \\
&= \frac{\SC(w | \yy)  + \delta(x_i, w)}{\SC(o | \yy)  + \delta(x_i, w)}.
\end{align*}
Now, we can use the inequality
\begin{align}\label{ineq}
\frac{\alpha + \gamma}{\beta + \gamma} \leq \frac{\alpha}{\beta},
\end{align}
which holds for every $\alpha, \beta, \gamma$ such that $\alpha \geq \beta$ and $\gamma \geq 0$. In particular, since $\SC(o | \xx) \leq \SC(w | \xx)$, by setting $\alpha = \SC(w | \xx)$, $\beta = \SC(o | \xx)$, $\gamma= \delta(x_i, w) \geq 0$, we obtain that
\begin{align*}
\frac{\SC(w | \xx)}{\SC(o | \xx)} \leq \frac{\SC(w | \yy)}{\SC(o | \yy)}.
\end{align*}

For (T2), consider an intermediate instance with position profile $\xx$ and let $z \in \calA$ be a location such that $z \leq o$ (the case $z \geq o$ can be handled similarly) which is representative for a set of districts $\calD_z \neq \varnothing$ and there exists an agent $i$ in $\calD_z$ positioned at $x_i \neq [z,o]$. We will show that the distortion of the instance obtained by moving only agent $i$ to the boundaries of $[z,o]$ is at least as large as that of the previous instance. We distinguish between two cases:
\begin{itemize}
\item $x_i < z$. In this case, agent $i$ is moved to $z$ yielding an instance with position profile $\yy = (z,\xx_{-i})$.
By the fact that (T1) is performed before (T2) and the assumption of this case, we have that $x_i \in [w,z)$. Hence, 
$$\delta(x_i,w) = x_i - w = z-w - (z-x_i) = \delta(z,w) - \delta(x_i,z) \leq \delta(z,w).$$ 
Furthermore, since $x_i \leq o$, we also have that 
$$\delta(x_i,o) = o-x_i = o-z+z-x_i = \delta(z,o) + \delta(x_i,z) \geq \delta(z,o).$$ 
Consequently, 
\begin{align*}
\frac{\SC(w | \xx)}{\SC(o | \xx)} 
&= \frac{\delta(x_i,w) + \sum_{j \neq i} \delta(x_j,w)}{ \delta(x_i,o) + \sum_{j \neq i} \delta(x_j,o)} \\
&\leq \frac{\delta(z,w) + \sum_{j \neq i} \delta(x_j,w)}{\delta(z,o) + \sum_{j \neq i} \delta(x_j,o)} \\
&= \frac{\SC(w | \yy)}{\SC(o | \yy)}.
\end{align*}

\item $x_i > o$. Now, agent $i$ is moved to $o$ yielding an instance with position profile $\yy = (o,\xx_{-i})$.
Since $o > w$, we have that 
$$\delta(x_i,w) = x_i - w = x_i - o + o-w =\delta(x_i,o) + \delta(o,w).$$
Hence, we have
\begin{align*}
\frac{\SC(w | \xx)}{\SC(o | \xx)} 
&= \frac{\delta(x_i,w) + \sum_{j \neq i} \delta(x_j,w)}{\delta(x_i,o) + \sum_{j \neq i} \delta(x_j,o)} \\
&= \frac{\delta(x_i,o) + \delta(o,w) + \sum_{j \neq i} \delta(x_j,w)}{\delta(x_i,o) + \delta(o,o) + \sum_{j \neq i} \delta(x_j,o)} \\
&= \frac{\SC(w | \yy) + \delta(x_i,o) }{\SC(o | \yy) + \delta(x_i,o) } .
\end{align*}
Now, since $o$ is still the optimal location, we have that $\SC(o | \yy) \leq \SC(w | \yy)$, and by applying \eqref{ineq} with $\alpha = \SC(w|\yy)$, $\beta=\SC(o|\yy)$ and $\gamma = \delta(x_i,o)$, we finally obtain 
\begin{align*}
\frac{\SC(w | \xx)}{\SC(o | \yy)} 
\leq \frac{\SC(w | \xx)}{\SC(o | \yy)}.
\end{align*}
\end{itemize}
This completes the proof of the lemma.
\end{proof}

\subsection{Bounding the distortion}\label{sec:discrete-distortion}
Given the class of instances $\wc(\calM)$ for any $\calM \in \{\dmm,\ddm\}$ and the characterization of \cref{lem:discrete-structure}, we are ready to bound the distortion of both mechanisms. 
Before we dive into the analysis, we present some useful notation, which will be used throughout the proofs presented in this section.
Consider any instance $\calI = (\xx,\calD,\calA) \in \wc(\calM)$. Due to symmetry, it will suffice to consider only the case where $\calM(\calI) = w < o = \opt(\calI)$. For every alternative location $z \in \calA$, let $\calD_z$ be the set of districts for which $z$ is the representative; that is, $z_d = z$ for every $d \in \calD_z$. Also, let $Z = \{z \in \calA: \calD_z \neq \varnothing \}$ be the set of all alternative locations that are representative for at least one district. Observe that since the location $w$ is selected by the mechanism, it must be the case that $w \in Z$. 
For every $z \in Z$ and $y \in \calA$, let 
\begin{align*}
\SC_z(y | \xx) = \sum_{d \in \calD_z} \sum_{i \in \calN_d} \delta(x_i, y)
\end{align*}
be the total distance of all the agents in the districts of $\calD_z$ from $y$.  Also, recall that each district contains exactly $\lambda$ agents.

\begin{theorem}\label{thm:DMM-distortion}
The distortion of $\dmm$ is at most $3$.
\end{theorem}

\begin{proof}
Consider any instance $\calI = (\xx,\calD,\calA) \in \wc(\dmm)$ with $\dmm(\calI) = w < o = \opt(\calI)$.
We make the following observations:
\begin{itemize}
\item 
Consider any alternative location $z \in \calA$ such that $\calD_z \neq \varnothing$. By property (P2), for any district $d \in \calD_z$, we have that $\delta(z,o) = \delta(x_i,z) + \delta(x_i,o)$ for every agent $i \in \calN_d$. Hence, by summing over all agents in the districts of $\calD_z$, we have 
\begin{align*}
\SC_z(z | \xx)  + \SC_z(o | \xx) = \delta(z,o) \cdot \lambda |\calD_z|.
\end{align*}
Since $z$ is chosen as the representative of each district $d \in \calD_z$, it is the location that minimizes the total distance of the agents in $d$, that is, $\sum_{i \in \calN_d} \delta(x_i,z) \leq \sum_{i \in \calN_d} \delta(x_i,o)$. Thus, by summing over all districts in $\calD_z$, we have that
\begin{align*}
\SC_z(z | \xx) \leq \SC_z(o | \xx).
\end{align*} 
By combining the above two expressions, we obtain 
\begin{align}
\SC_z(z | \xx)  &=  \frac{1}{2} \delta(z,o) \cdot \lambda |\calD_z|. \label{eq:DMM-zz}
\end{align}
and 
\begin{align}
\SC_z(o | \xx) &\geq \frac{1}{2} \delta(z,o) \cdot \lambda |\calD_z|. \label{eq:DMM-zo}
\end{align}

\item   
Consider any alternative location $z \in Z \setminus\{w\}$. 
By property (P1), we have that $w$ is the left-most representative, and thus $z > w$.
By (P2), we have that every agent $i$ in a district of $\calD_z$ lies in the interval defined by $z$ and $o$, which means that 
\begin{itemize}
\item $\delta(x_i,w) \leq \delta(w,o)$ if $z \leq o$, and 
\item $\delta(x_i,w) \leq \delta(w,z) = \delta(w,o) + \delta(z,o)$ if $z > o$. 
\end{itemize}
Since $\delta(z,o) \geq 0$, by summing over all the agents in the districts of $\calD_z$, we obtain that
\begin{align}\label{eq:DMM-zw}
\SC_z(w | \xx) &\leq \bigg( \delta(w,o) + \delta(z,o) \bigg) \cdot \lambda |\calD_z|.
\end{align}

\item
Since $w$ is the left-most representative (implied by (P1)) and the median among all representatives (which is why it is selected by the mechanism), it must be the case that $w$ is the representative of more than half of the districts, and thus 
\begin{align}\label{eq:DMM-wz-size}
|\calD_w| \geq \sum_{z \in Z\setminus\{w\}} |\calD_z|.
\end{align} 
\end{itemize}

Given the above observations, we will now upper-bound the social cost of $w$ and lower-bound the social cost of $o$ (in order to obtain an upper bound on the distortion of $\calI$ subject to using $\dmm$). 
By the definition of $\SC(w | \xx)$, and by applying \eqref{eq:DMM-zz} for $y=w$ and \eqref{eq:DMM-zw} for $z \neq w$, 
we obtain
\begin{align*}
\SC(w|\xx) &= \SC_w(w | \xx) + \sum_{z \in Z \setminus\{w\}} \SC_z(w|\xx) \\
&\leq \frac{1}{2} \delta(w,o) \cdot \lambda |\calD_w| 
+ \sum_{z \in Z \setminus\{w\}} \bigg( \delta(w,o) + \delta(z,o) \bigg) \cdot \lambda |\calD_z|  \\
&= \frac{1}{2} \delta(w,o) \cdot \lambda |\calD_w| 
+ \delta(w,o) \cdot \lambda \sum_{z \in Z \setminus\{w\}} |\calD_z| 
+  \sum_{z \in Z \setminus\{w\}}  \delta(z,o) \cdot \lambda  |\calD_z|.
\end{align*}
By \eqref{eq:DMM-wz-size}, we further have that
\begin{align}
\SC(w | \xx) &\leq \frac{3}{2} \delta(w,o) \cdot \lambda |\calD_w| 
+  \sum_{z \in Z \setminus\{w\}}  \delta(z,o) \cdot \lambda  |\calD_z|  \nonumber \\
&\leq \frac{3}{2} \sum_{z \in Z}  \delta(z,o) \cdot \lambda  |\calD_z|.  \label{eq:DMM-mech}
\end{align}
On the other hand, by the definition of $\SC(o | \xx)$ and by applying \eqref{eq:DMM-zo}, we can lower-bound the optimal social cost as follows:
\begin{align}\label{eq:DMM-opt}
\SC(o | \xx) &= \sum_{z \in Z} \SC_z(o|\xx) \geq \frac{1}{2} \sum_{z \in Z}  \delta(z,o) \cdot \lambda  |\calD_z|.
\end{align}
Consequently, by combining \eqref{eq:DMM-mech} and \eqref{eq:DMM-opt}, the distortion of the instance $\calI$ subject to $\dmm$ is
\begin{align*}
\dist(\calI | \ddm) = \frac{\SC(w | \xx)}{\SC(o | \xx)} 
\leq 3.
\end{align*}
Since $\calI$ is an arbitrary (up to symmetry) instance of $\wc(\ddm)$, \cref{lem:discrete-structure} implies $\dist(\ddm) \leq 3$.
\end{proof}

Next, we bound the distortion of $\ddm$.

\begin{theorem}\label{thm:DDM-distortion}
The distortion of $\ddm$ is at most $7$.
\end{theorem}

\begin{proof}
Consider any instance $\calI = (\xx,\calD,\calA) \in \wc(\ddm)$ with $\ddm(\calI) = w < o = \opt(\calI)$.
We make the following observations:
\begin{itemize}
\item 
Consider any district $d \in \calD_z$. By property (P2), we have that $\delta(z,o) = \delta(x_i,z) + \delta(x_i,o)$ for every agent $i \in \calN_d$.
Furthermore, by combining the fact that all agents in $\calN_d$ lie in the interval defined by $z$ and $o$, together with the fact that the median agent of $d$ is closer to $z$ than to $o$, we have that there exists a set $S_d \subseteq \calN_d$ of agents in $d$ with $|S_d| \geq \frac{1}{2}\lambda$ such that $\delta(x_i,z) \leq \delta(x_i,o)$ for every $i \in S_d$.  Consequently, $\delta(x_i,z) \leq \frac{1}{2}\delta(z,o)$ for every $i \in S_d$, and $\delta(x_i,z) \leq \delta(z,o)$ for every $i \in \calN_d \setminus S_d$. We obtain 
\begin{align}
\SC_z(z | \xx) &= \sum_{d \in \calD_z} \sum_{i \in \calN_d} \delta(x_i, z) \nonumber \\
&= \sum_{d \in \calD_z} \bigg( \sum_{i \in S_d} \delta(x_i, z) + \sum_{i \in \calN_d \setminus S_d} \delta(x_i, z)\bigg) \nonumber \\
&\leq \sum_{d \in \calD_z} \bigg( \frac{1}{2} \delta(z,o) \cdot |S_d| + \delta(z,o) \cdot |\calN_d \setminus S_d| \bigg) \nonumber \\
&= \delta(z,o) \cdot \sum_{d \in \calD_z} \bigg( \lambda - \frac{1}{2}|S_d| \bigg) \nonumber \\
&\leq \frac{3}{4} \delta(z,o) \cdot \lambda |\calD_z|. \label{eq:DDM-zz}
\end{align}

\item 
By summing the equality $\delta(x_i,o) = \delta(z,o) - \delta(x_i,z)$ over all agents $i$ in the districts of $\calD_z$ and using \eqref{eq:DDM-zz}, we obtain
\begin{align}
\SC_z(o | \xx)  &=  \delta_{z,o} \cdot \lambda |\calD_z| - \SC_z(z | \xx) \nonumber \\
&\geq \frac{1}{4} \delta(z,o) \cdot \lambda |\calD_z|. \label{eq:DDM-zo}
\end{align}

\item Using the same arguments as in the proof of \cref{thm:DMM-distortion} for $\dmm$, we can show that inequalities \eqref{eq:DMM-zw} and \eqref{eq:DMM-wz-size} hold for $\ddm$ as well. 
\end{itemize}
Using these observations, we can now upper-bound the social cost of $w$ and lower-bound the social cost of $o$. 
By the definition of $\SC(w | \xx)$, and by applying \eqref{eq:DDM-zz} for $y=w$ and \eqref{eq:DMM-zw} for $z \neq w$, 
we obtain
\begin{align*}
\SC(w|\xx) &= \SC_w(w | \xx) + \sum_{z \in Z \setminus\{w\}} \SC_z(w|\xx) \\
&\leq \frac{3}{4} \delta(w,o) \cdot \lambda |\calD_w| 
+ \sum_{z \in Z \setminus\{w\}} \bigg( \delta(w,o) + \delta(z,o) \bigg) \cdot \lambda |\calD_z|  \\
&= \frac{3}{4} \delta(w,o) \cdot \lambda |\calD_w| 
+ \delta(w,o) \cdot \lambda \sum_{z \in Z \setminus\{w\}} |\calD_z| 
+  \sum_{z \in Z \setminus\{w\}}  \delta(z,o) \cdot \lambda  |\calD_z|. 
\end{align*}
By \eqref{eq:DMM-wz-size}, we further have that
\begin{align}
\SC(w | \xx) &\leq \frac{7}{4} \delta(w,o) \cdot \lambda |\calD_w| 
+  \sum_{z \in Z \setminus\{w\}}  \delta(z,o) \cdot \lambda  |\calD_z|  \nonumber \\
&\leq \frac{7}{4} \sum_{z \in Z}  \delta(z,o) \cdot \lambda  |\calD_z|.  \label{eq:DDM-mech}
\end{align}
On the other hand, by the definition of $\SC(o | \xx)$ and by applying \eqref{eq:DDM-zo}, we can lower-bound the optimal social cost as follows:
\begin{align}\label{eq:DDM-opt}
\SC(o | \xx) &= \sum_{z \in Z} \SC_z(o|\xx) \geq \frac{1}{4} \sum_{z \in Z}  \delta(z,o) \cdot \lambda  |\calD_z|.
\end{align}
Consequently, by combining \eqref{eq:DDM-mech} and \eqref{eq:DDM-opt}, the distortion of the instance $\calI$ is
\begin{align*}
\dist(\calI | \ddm) = \frac{\SC(w | \xx)}{\SC(o | \xx)} 
\leq 7.
\end{align*}
Since $\calI$ is an arbitrary (up to symmetry) instance of $\wc(\ddm)$, \cref{lem:discrete-structure} implies $\dist(\ddm) \leq 7$.
\end{proof}

%%%%%%
%%%%%%

\section{Lower bounds on the distortion for the discrete setting}\label{sec:discrete-lower}

In this section, we present our lower bounds for the discrete setting. 
Specifically we show the following two statements:
\begin{itemize}
\item The distortion of any mechanism is at least $3-\varepsilon$, for any $\varepsilon >0$.
\item The distortion of any strategyproof mechanism is at least $7-\varepsilon$, for any $\varepsilon >0$.
\end{itemize}
These lower bounds match the upper bounds presented in \cref{sec:discrete}. 
Consequently, $\DMM$ is the best-possible among all mechanisms (in terms of distortion), while $\DDM$ is the best-possible strategyproof mechanism. 

Before we dive into the proofs of our lower bounds, we remark that it is without loss of generality to assume that, when given as input an instance with only two districts each of which has a different representative, any mechanism will choose the left-most representative as the facility location; if this is not the case, then we can obtain the very same bounds by symmetric arguments. Furthermore, in this section, as well as in \cref{sec:continuous-lower}, we will simply write $\SC(z)$ instead of $\SC(z|\xx)$ for the social cost of an alternative location $z$; the position profile $\xx$ will always be clear from context.

\subsection{An unconditional lower bound}\label{sec:discrete-lower-unconditional}

First, we will prove a general lemma about mechanisms that have approximation ratio less than $3-\varepsilon$, for any $\varepsilon > 0$. 

\begin{lemma}\label{lem:011winner-general}
Let $\calM$ be a mechanism with distortion strictly less than $3$. 
Let $\calI$ be an instance with set of alternative locations $\calA = \{0,1\}$, and $k=2\mu+1$ districts such that $0$ is the representative of $\mu$ districts and $1$ is the representative of $\mu+1$ districts, for every integer $\mu \geq 1$.
Then, 
\begin{itemize}
\item[(i)] $\calM(\calI) = 1$, and 
\item[(ii)] the representative of any district $d$ for which all agents are positioned at $\frac{2\mu+1}{4(\mu+1)}$ is $z_d = 0$.	
\end{itemize}
\end{lemma}

\begin{proof}
We will prove the statement by induction on $\mu$.

\bigskip
\noindent
\underline{Base case: $\mu=1$.} \\
For (i), assume towards a contradiction that there exists an instance $\calI$ such that $0$ is the representative of one district, and $1$ is the representative of two districts, but $\calM(\calI) = 0$. In particular, let $\calI$ be the following instance with three districts:
\begin{itemize}
\item In the first district, all $\lambda$ agents are positioned at $1/4$. The representative of this district must be $0$, as otherwise the distortion of the instance consisting only this district would be $3$.

\item In the other two districts, all $\lambda$ agents are positioned at $1$. Since $\calM$ is unanimous, the representative of this district is $1$.
\end{itemize}
See the left part of \cref{fig:lower-general-basecase} for a graphical representation of $\calI$. 
Since 
$$\SC(0) = \frac{\lambda}{4}+2\lambda = \frac{9\lambda}{4} \text{ \ \ and \ \ } \SC(1) = \frac{3\lambda}{4},$$ 
we have that $\dist(\calM) \geq \dist(\calI | \calM) = 3$, a contradiction. 

\medskip

For (ii), assume towards a contradiction that the representative of the district in which all agents are positioned at 
$\frac{2\mu+1}{4(\mu+1)} = 3/8$ is $1$.  Consider the following instance $\calJ$ with three districts:
\begin{itemize}
\item In the first district, all $\lambda$ agents are positioned at $0$. 
Since $\calM$ is unanimous, the representative of this district is $0$.

\item In the other two districts, all $\lambda$ agents are positioned at $3/8$. 
By our assumption, the representative of these two districts is $1$.
\end{itemize}
The right part of \cref{fig:lower-general-basecase} depicts instance $\calJ$. 
Since $\calJ$ satisfies the properties of the lemma, by (i), it must be $\calM(\calJ) = 1$. 
However, since 
$$\SC(0) = \frac{6\lambda}{8} = \frac{3\lambda}{4}  \text{ \ \ and \ \ } \SC(1) =\lambda+\frac{10 \lambda}{8} = \frac{9\lambda}{4},$$ 
we again have that $\dist(\calM) \geq \dist(\calJ | \calM) = 3$, a contradiction. 

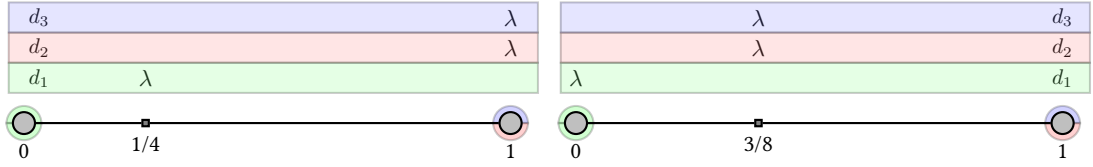
\begin{figure}
	\centering
	\begin{minipage}[b]{0.45\linewidth}
		\centering
		\begin{tikzpicture}[thick,scale=0.8, every node/.style={scale=0.8}]
		\node [semicircle, draw=black,fill=white!10!green,opacity=0.2] (x) at (-4, 0.125) {};
		\node [semicircle, draw=black,rotate=180,fill=white!10!green,opacity=0.2] (x) at (-4, -0.125) {};
		\node [circle, draw=black,color=black,fill=white!50!gray, label={below:0}] (0) at (-4, 0) {};
		\node [semicircle, draw=black,fill=white!10!blue,opacity=0.2] (x) at (4, 0.125) {};
		\node [semicircle, draw=black,rotate=180,fill=white!10!red,opacity=0.2] (x) at (4, -0.125) {};
		\node [circle, draw=black,color=black,fill=white!50!gray, label={below:1}] (1) at (4, 0) {};
		\node [draw,fill=gray,inner sep=1.5pt, label={below:1/4}] (1/4) at (-2, 0) {};
		\draw (0) to (1/4);
		\draw (1/4) to (1);
		\node (district1) at (-2,0.75) {$\lambda$};
		\node at (-3.75,0.75) {\small{$d_1$}};
		\draw[draw=black,fill=white!50!green,opacity=0.2] (-4.25,0.5) rectangle ++(8.7,0.5);
		\node (district2) at (4,1.25) {$\lambda$};
		\node at (-3.75,1.25) {\small{$d_2$}};
		\draw[draw=black,fill=white!50!red,opacity=0.2] (-4.25,1) rectangle ++(8.7,0.5);
		\node (district3) at (4,1.75) {$\lambda$};
		\node at (-3.75,1.75) {\small{$d_3$}};
		\draw[draw=black,fill=white!50!blue,opacity=0.2] (-4.25,1.5) rectangle ++(8.7,0.5);
		\end{tikzpicture}		
	\end{minipage}
	\begin{minipage}[b]{0.45\linewidth}
		\centering
		\begin{tikzpicture}[thick,scale=0.8, every node/.style={scale=0.8}]		
		\node [semicircle, draw=black,fill=white!10!green,opacity=0.2] (x) at (-4, 0.125) {};
		\node [semicircle, draw=black,rotate=180,fill=white!10!green,opacity=0.2] (x) at (-4, -0.125) {};
		\node [circle, draw=black,color=black,fill=white!50!gray, label={below:0}] (0) at (-4, 0) {};
		\node [semicircle, draw=black,fill=white!10!blue,opacity=0.2] (x) at (4, 0.125) {};
		\node [semicircle, draw=black,rotate=180,fill=white!10!red,opacity=0.2] (x) at (4, -0.125) {};
		\node [circle, draw=black,color=black,fill=white!50!gray, label={below:1}] (1) at (4, 0) {};
		\node [draw,fill=gray,inner sep=1.5pt, label={below:3/8}] (3/8) at (-1, 0) {};
		\draw (0) to (3/8);
		\draw (3/8) to (1);
		\node (district1) at (-4,0.75) {$\lambda$};
		\node at (4,0.75) {\small{$d_1$}};
		\draw[draw=black,fill=white!50!green,opacity=0.2] (-4.25,0.5) rectangle ++(8.7,0.5);
		\node (district2) at (-1,1.25) {$\lambda$};
		\node at (4,1.25) {\small{$d_2$}};
		\draw[draw=black,fill=white!50!red,opacity=0.2] (-4.25,1) rectangle ++(8.7,0.5);
		\node (district3) at (-1,1.75) {$\lambda$};
		\node at (4,1.75) {\small{$d_3$}};
		\draw[draw=black,fill=white!50!blue,opacity=0.2] (-4.25,1.5) rectangle ++(8.7,0.5);
		\end{tikzpicture}
	\end{minipage}
	\caption{The base case of \cref{lem:011winner-general}. 
	\emph{Left:} The proof of part (i). 
	In instance $\calI$, the representative of district $d_1$ is $0$ (shaded green), as otherwise the distortion of the instance consisting only of $d_1$ would be $3$. By unanimity, the representative of $d_2$ (shaded red) and $d_3$ (shaded blue) is $1$ (shaded both red and blue). 
	\emph{Right:} The proof of part (ii). 
	In instance $\calJ$, by unanimity, the representative of $d_1$ is $0$ (shaded green). We assume towards a contradiction that the representative of $d_2$ (shaded red) and $d_3$ (shaded blue) is $1$ (shaded both red and blue), and obtain a distortion of at least $3$.}
	\label{fig:lower-general-basecase}
\end{figure}

\bigskip
\noindent 
\underline{Induction step:} \\
We assume that (i) and (ii) are true for $\mu=\ell-1$, and will show that they are also true for $\mu=\ell$.

\medskip

\noindent 
For (i), consider an instance $\calI$ with $2\ell+1$ districts, such that $0$ is the representative of $\ell$ of them, 
and $1$ is the representative of the remaining $\ell+1$ districts. 
Specifically:
\begin{itemize}
\item In each of the first $\ell$ districts, all $\lambda$ agents are positioned at $\frac{2\ell-1}{4\ell} = \frac{2(\ell-1)+1}{4((\ell-1)+1)}$. 
By part (ii) of the induction hypothesis, the representative of all these districts is $0$.

\item In each of the remaining $\ell+1$ districts, all $\lambda$ agents are positioned at $1$. 
Since $\calM$ is unanimous, the representative of these districts is $1$. 
\end{itemize}
We now have that
$$\SC(0) = \ell \cdot \frac{(2\ell-1)\lambda}{4\ell} + (\ell+1)\cdot \lambda = \frac{3(2\ell+1)\lambda}{4}$$ 
and
$$\SC(1) = \ell \cdot \frac{(2\ell+1)\lambda}{4\ell} = \frac{(2\ell+1)\lambda}{4}.$$ 
If $\calM(\calI) = 0$ then $\dist(\calM) \geq \dist(\calI | \calM) = 3$.
Therefore, for the mechanism to achieve distortion strictly less than $3$, it must be the case that $\calM(\calI) = 1$.

For (ii), assume towards a contradiction that the representative of a district in which all agents are positioned at $\frac{2\ell+1}{4(\ell+1)}$ is $1$ instead. 
Then, consider the following instance $\calJ$ with $2\ell+1$ districts:
\begin{itemize}
\item In each of the first $\ell$ districts, all $\lambda$ agents are positioned at $0$. 
By unanimity, the representative of these districts is $0$.

\item In each of the remaining $\ell+1$ districts, all $\lambda$ agents are positioned at $\frac{2\ell+1}{4(\ell+1)}$. 
By our assumption, the representative of these districts is $1$.
\end{itemize}
Since (i) holds for $\mu=\ell$, it must be $\calM(\calJ) = \beta$. 
However, since  
$$\SC(0) = (\ell+1) \cdot \frac{(2\ell+1)\lambda}{4(\ell+1)} = \frac{\lambda(2\ell+1)}{4}$$ 
and 
$$\SC(1) = \ell \cdot \lambda + (\ell+1) \cdot \frac{(2\ell+3)\lambda}{4(\ell+1)} = \frac{3\lambda(2\ell+1)}{4},$$ 
we have that $\dist(\calM) \geq \dist(\calJ | \calM) = 3$, a contradiction.

\bigskip

\noindent 
This concludes the proof of the lemma.
\end{proof}

\noindent 
We are now ready to prove the main theorem.

\begin{theorem}\label{thm:discrete-lower-unconditional}
In the discrete setting, the distortion of any mechanism is at least $3-\varepsilon$, for any $\varepsilon > 0$.
\end{theorem}

\begin{proof}
Let $\calM$ be any mechanism with distortion less than $3-\varepsilon$, for any $\varepsilon > 0$. 
Hence, by \cref{lem:unanimous}, $\calM$ is unanimous within districts. 
We consider instances with set of alternative locations $\calA = \{0,1\}$. 
We will first establish that $\calM$ must choose $1$ as the representative of any district in which all the agents are positioned at $1/2$. 
To see this, consider the following instance $\calI$ with two districts:
\begin{itemize}
\item In the first district, all $\lambda$ agents are positioned at $1/2$.

\item In the second district, all $\lambda$ agents are positioned at $1$. 
Since $\calM$ is unanimous, the representative of this district is $1$.
\end{itemize}	
See also the left part of \cref{fig:lower-general}. We claim that the representative of the first district must be $1$. 
Assume towards a contradiction that the representative of this district is $0$.
Then, we have one district with $0$ as its representative and one district with $1$ as its representative. 
Recall that in such a case it is without loss of generality to assume that $\calM$ will select the left-most district representative, that is, $\calM(\calI) = 0$.
However, since 
$$\SC(0) = \frac{\lambda}{2}+\lambda = \frac{3\lambda}{2} \text{\ \ and \ \ } \SC(1) = \frac{\lambda}{2},$$
this decision leads to $\dist(\calM) \geq \dist(\calI | \calM) = 3$, a contradiction.

Finally, consider the following instance $\calJ$ with $k=2\mu+1$ districts (see the right part of \cref{fig:lower-general}):
\begin{itemize}
\item In each of the first $\mu$ districts, all $\lambda$ agents are positioned at $0$. 
By unanimity, the representative of these districts is $0$.

\item In each of the remaining $\mu+1$ districts, all $\lambda$ agents are positioned at $1/2$. 
By the above discussion, the representative of these districts is $1$.
\end{itemize}
By \cref{lem:011winner-general}, we have that $\calM(\calJ) = 1$. 
Observe that 
$$\SC(0) = (\mu+1)\cdot \frac{\lambda}{2}$$ 
and 
$$\SC(1) = \mu \cdot  \lambda + (\mu+1) \cdot \frac{\lambda}{2} = \frac{(3\mu+1)\lambda}{2}.$$ 
Hence, $\dist(\calJ | \calM) = \frac{3\mu+1}{\mu+1}$. 
By choosing $\mu$ to be sufficiently large, we obtain $\dist(\calM) \geq 3-\varepsilon$, for any $\varepsilon > 0$. 

\begin{figure}
	\centering
	\begin{minipage}[b]{0.45\linewidth}
		\centering
		\begin{tikzpicture}[thick,scale=0.8, every node/.style={scale=0.8}]
		\node [circle, draw=black,color=black,fill=white!50!gray, label={below:0}] (0) at (-4, 0) {};
		\node [semicircle, draw=black,fill=white!10!red,opacity=0.2] (x) at (4, 0.125) {};
		\node [semicircle, draw=black,rotate=180,fill=white!10!green,opacity=0.2] (x) at (4, -0.125) {};
		\node [circle, draw=black,color=black,fill=white!50!gray, label={below:1}] (1) at (4, 0) {};
		\node [draw,fill=gray,inner sep=1.5pt, label={below:1/2}] (1/2) at (0, 0) {};
		\draw (0) to (1/2);
		\draw (1/2) to (1);
		\node (district1) at (0,0.75) {$\lambda$};
		\node at (-3.75,0.75) {\small{$d_1$}};
		\draw[draw=black,fill=white!50!green,opacity=0.2] (-4.25,0.5) rectangle ++(8.7,0.5);
		\node (district2) at (4,1.25) {$\lambda$};
		\node at (-3.75,1.25) {\small{$d_2$}};
		\draw[draw=black,fill=white!50!red,opacity=0.2] (-4.25,1) rectangle ++(8.7,0.5);
		\end{tikzpicture}		
	\end{minipage}
	\begin{minipage}[b]{0.45\linewidth}
		\centering
		\begin{tikzpicture}[thick,scale=0.8, every node/.style={scale=0.8}]
		\node [semicircle, draw=black,fill=white!10!blue,opacity=0.2] (x) at (-4, 0.125) {};
		\node [semicircle, draw=black,rotate=180,fill=white!10!blue,opacity=0.2] (x) at (-4, -0.125) {};
		\node [circle, draw=black,color=black,fill=white!50!gray, label={below:0}] (0) at (-4, 0) {};		
		\node [semicircle, draw=black,fill=white!10!green,opacity=0.2] (x) at (4, 0.125) {};
		\node [semicircle, draw=black,rotate=180,fill=white!10!green,opacity=0.2] (x) at (4, -0.125) {};
		\node [circle, draw=black,color=black,fill=white!50!gray, label={below:1}] (1) at (4, 0) {};
		\node [draw,fill=gray,inner sep=1.5pt, label={below:1/2}] (1/2) at (0, 0) {};
		\draw (0) to (1/2);
		\draw (1/2) to (1);
		\node (district1) at (-4,0.75) {$\lambda$};
		\node at (3.5,0.75) {\small{$d_1$ ... $d_\mu$}};
		\draw[draw=black,fill=white!50!blue,opacity=0.2] (-4.25,0.5) rectangle ++(8.7,0.5);
		\node (district2) at (0,1.25) {$\lambda$};
		\node at (3.3,1.25) {\small{$d_{\mu+1}$ ... $d_{2\mu+1}$}};
		\draw[draw=black,fill=white!50!green,opacity=0.2] (-4.25,1) rectangle ++(8.7,0.5);
		\end{tikzpicture}
	\end{minipage}
	\caption{The instances used in the proof of \cref{thm:discrete-lower-unconditional}. 
	Instance $\calI$ ({\em left}): By unanimity, the representative of district $d_2$ (shaded red) is $1$. 
	The representative of $d_1$ (shaded green) must also be $1$, as otherwise $\dist(\calI|\calM)=3$. 
	Instance $\calJ$ ({\em right}): Since districts $d_{\mu+1}$ to $d_{2\mu+1}$ are identical to $d_1$ in $\calI$, the representative of those districts must be $1$ (shaded green). 
	By unanimity, the representative of $d_1$ to $d_\mu$ is $0$ (shaded blue). 
	By \cref{lem:011winner-general}, the facility location on $\calJ$ is $1$, and thus $\dist(\calJ|\calM) = 3$.}
	\label{fig:lower-general}
\end{figure}
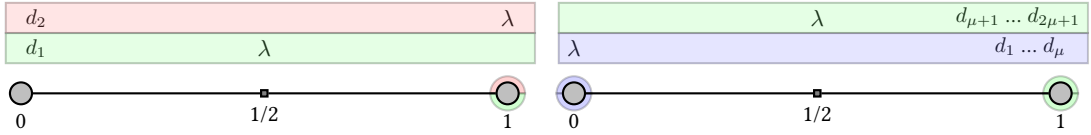

\bigskip
\noindent This concludes the proof.	
\end{proof}

\subsection{A lower bound for strategyproof mechanisms}\label{sec:discrete-lower-sp}

The following lemma will be very important for the lower bound. It establishes that a strategyproof mechanism is essentially \emph{ordinal}. 
That is, the outcome does not depend on the exact positions of the agents, but only on the preference orderings over the alternative locations that those positions induce. 

\begin{lemma}\label{lem:spordinal}
Let $\calM$ be a strategyproof mechanism with finite distortion. 
Let $\xx_d$ and $\yy_d$ be two different district position profiles for some district $d$, 
such that for every agent $i \in \calN_d$ and any two alternative locations $\alpha \neq \beta$,
\begin{itemize}
\item $\delta(x_i,\alpha) \neq \delta(x_i, \beta)$ and $\delta(y_i,\alpha) \neq \delta(y_i, \beta)$;
\item If $\delta(x_i,\alpha) < \delta(x_i, \beta)$ then $\delta(y_i,\alpha) < \delta(y_i, \beta)$.
\end{itemize}
Then, the representative of the district chosen by $\calM$ is the same under both $\xx_d$ and $\yy_d$. 
\end{lemma}

Next, we will show a lemma similar to \cref{lem:011winner-general}. 
Before we continue, we remark that in the proofs of this section we will construct sets of instances with one or more districts. 
In any instance with a single district, its representative must necessarily also be the facility location. For ease of notation, in case of a single-district instance $\calI$, we will use $\calI$ to denote both the instance and the (single) district of the instance.

\begin{lemma}\label{lem:011winner}
Let $\calM$ be a strategyproof mechanism with distortion less than $7-\varepsilon$, for any $\varepsilon > 0$.
Let $\calI$ be an instance with set of alternative locations $\calA = \{0,1\}$, and $k=2\mu+1$ districts such that $0$ is the representative of $\mu$ districts and $1$ is the representative of $\mu+1$ districts, for every integer $\mu \geq 1$.
Then, $\calM(\calI) = 1$. 	
\end{lemma}

\begin{proof}	
We will prove the statement by induction on $\mu$. 

\bigskip
\noindent 
\underline{Base case: $\mu=1$.} \\ 
We will first define a particularly structured single-district instance $\calI_1$ and will argue that $\calM(\calI_1) = 0$. Next, we will define a second
single-district instance $\calI_2$ and, using the property $\calM(\calI_1) = 0$, we will show that $\calM(\calI_2) = 0$ as well. 
Finally, we will define a third instance $\calI_3$ consisting of (the district of) $\calI_2$ and two other districts, in which $1$ will be the unanimous representative. By using the structure of $\calI_2$ and the property $\calM(\calI_2) = 0$, we will argue that it must be $\calM(\calI_3) =1$.
Let $\delta$ be a parameter which can become infinitesimally small (tends to zero). 

\medskip
\noindent 
\textbf{Instance $\calI_1$:} 
There is a single district in which $3\lambda/4$ agents are positioned at $0$ and $\lambda/4$ agents are positioned at $1/2+\delta$. 
Hence, 
$$\SC(0) = \frac{\lambda}{4} \cdot \left(\frac{1}{2}+\delta\right) = \frac{\lambda}{8}+\frac{\lambda \delta}{4},$$ 
and
$$\SC(1) = \frac{\lambda}{4} \cdot \left(\frac{1}{2}-\delta\right) + \frac{3\lambda}{4} = \frac{7\lambda}{8}-\frac{\lambda \delta}{4}.$$ 
By taking $\delta$ to be sufficiently small, it is easy to see that if $\calM(\calI_1)=1$ then $\dist(\calM) \geq 7-\varepsilon$, for any $\varepsilon >0$. Hence, it must be $\calM(\calI_1)=0$. Because the instance consists of only one district, we also have that the representative of this district must be $0$.

\medskip
\noindent 
\textbf{Instance $\calI_2$:} 
There is a single district in which $3\lambda/4$ agents are positioned at $1/2-\delta$ and $\lambda/4$ agents are positioned at $1$. 
Note that in both $\calI_1$ and $\calI_2$, the positions of the agents induce the same ordering over the alternative locations $0$ and $1$,  and there are no agents that are indifferent between the two alternative locations in any of the two instances. 
Since $\calM(\calI_1)= 0$, by \cref{lem:spordinal}, it must also be the case that $\calM(\calI_2) = 0$. See also the left part of \cref{fig:lower-sp-basecase}.

\medskip
\noindent 
\textbf{Instance $\calI_3$:} 
There are three districts which will be such that the representative of the first one is $0$ and the representative of the remaining two is $1$. In particular:
\begin{itemize}
\item The first district is identical to the district of $\calI_2$: $3\lambda/4$ agents are positioned at $1/2-\delta$ and $\lambda/4$ agents are positioned at $1$. By the above discussion, the representative of this district is $0$.

\item In the other two districts, all agents are positioned at $1$. 
By unanimity, the representative of both districts is $1$. 
\end{itemize}
We have 
$$\SC(0) = \frac{3\lambda}{4}\left(\frac{1}{2}-\delta\right) + \frac{\lambda}{4} + 2\lambda = \frac{21\lambda}{8} - \frac{3\lambda\delta}{4}$$ 
and
$$\SC(1) = \frac{3\lambda}{4}\left(\frac{1}{2}+\delta\right) = \frac{3\lambda}{8} + \frac{3\lambda\delta}{4}.$$ 
By taking $\delta$ to be sufficiently small, it is easy to see that if $\calM(\calI_3) = 0$ then $\dist(\calM) \geq 7-\varepsilon$, for any $\varepsilon >0$. Therefore, it must be $\calM(\calI_3)=1$, as desired. See also the right part of \cref{fig:lower-sp-basecase}.

\begin{figure}
	\centering
	\begin{minipage}[b]{0.45\linewidth}
		\centering
		\begin{tikzpicture}[thick,scale=0.8, every node/.style={scale=0.8}]
		\node [semicircle, draw=black,fill=white!10!green,opacity=0.2] (x) at (-4, 0.125) {};
		\node [semicircle, draw=black,rotate=180,fill=white!10!DarkGreen,opacity=0.2] (x) at (-4, -0.125) {};
		\node [circle, draw=black,color=black,fill=white!50!gray, label={below:0}] (0) at (-4, 0) {};
		\node [circle, draw=black,color=black,fill=white!50!gray, label={below:1}] (1) at (4, 0) {};
		\node [draw,fill=gray,inner sep=1.5pt, label={below:1/2-$\delta$}] (1/2) at (-0.4, 0) {};
		\node [draw,fill=gray,inner sep=1.5pt, label={above:1/2+$\delta$}] (1/22) at (0.4, 0) {};
		\draw (0) to (1/2);
		\draw (1/2) to (1);
		\node (district1) at (-4,1) {$3\lambda/4$};
		\node (district1) at (0.4,1) {$\lambda/4$};
		\draw[draw=black,fill=white!50!DarkGreen,opacity=0.2] (-4.5,0.75) rectangle ++(9,0.5);
		\node (district2) at (4,1.75) {$\lambda/4$};
		\node (district2) at (-0.4,1.75) {$3\lambda/4$};
		\draw[draw=black,fill=white!50!green,opacity=0.2] (-4.5,1.5) rectangle ++(9,0.5);
		\end{tikzpicture}		
	\end{minipage}
	\begin{minipage}[b]{0.45\linewidth}
		\centering
		\begin{tikzpicture}[thick,scale=0.8, every node/.style={scale=0.8}]
		\node [semicircle, draw=black,fill=white!10!green,opacity=0.2] (x) at (-4, 0.125) {};
		\node [semicircle, draw=black,rotate=180,fill=white!10!green,opacity=0.2] (x) at (-4, -0.125) {};
		\node [circle, draw=black,color=black,fill=white!50!gray, label={below:0}] (0) at (-4, 0) {};
		\node [semicircle, draw=black,fill=white!10!blue,opacity=0.2] (x) at (4, 0.125) {};
		\node [semicircle, draw=black,rotate=180,fill=white!10!red,opacity=0.2] (x) at (4, -0.125) {};
		\node [circle, draw=black,color=black,fill=white!50!gray, label={below:1}] (1) at (4, 0) {};
		\node [draw,fill=gray,inner sep=1.5pt, label={below:1/2-$\delta$}] (1/2) at (-0.4, 0) {};
		\draw (0) to (1/2);
		\draw (1/2) to (1);
		\node (district1) at (-0.4,0.75) {$3\lambda/4$};
		\node (district1) at (4,0.75) {$\lambda/4$};
		\node at (-3.75,0.75) {\small{$d_1$}};
		\draw[draw=black,fill=white!50!green,opacity=0.2] (-4.25,0.5) rectangle ++(8.75,0.5);
		\node (district2) at (4,1.25) {$\lambda$};
		\node at (-3.75,1.25) {\small{$d_2$}};
		\draw[draw=black,fill=white!50!red,opacity=0.2] (-4.25,1) rectangle ++(8.75,0.5);
		\node (district3) at (4,1.75) {$\lambda$};
		\node at (-3.75,1.75) {\small{$d_3$}};
		\draw[draw=black,fill=white!50!blue,opacity=0.2] (-4.25,1.5) rectangle ++(8.75,0.5);
		\end{tikzpicture}
	\end{minipage}
	\caption{The base case of \cref{lem:011winner}. 
	\emph{Left:} The single-district instances $\calI_1$ (shaded in dark green) and $\calI_2$ (shaded in green). 
	The representative of the district of $\calI_1$ is $0$, as otherwise $\dist(\calI_1|\calM) \geq 7 - \varepsilon$. 
	By \cref{lem:spordinal}, $0$ is also the representative of the district of $\calI_2$. 
	\emph{Right:} The instance $\calI_3$. 
	District $d_1$ (shaded green) is identical to the district of $\calI_2$, and hence its representative is $0$. 
	By unanimity, the representative of $d_2$ (shaded red) and $d_3$ (shaded blue) is $1$ (shaded both red and blue). 
	The facility location on $\calI_3$ must be $1$, as otherwise $\dist(\calI_3|\calM)\geq 7-\varepsilon$.} 
	\label{fig:lower-sp-basecase}
\end{figure}
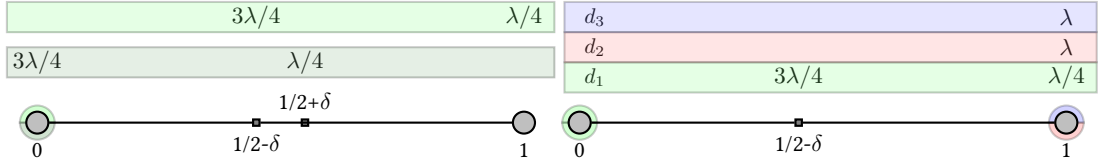

\bigskip
\noindent 
\underline{Induction step:} \\
Our induction hypothesis is that the lemma is true for $\mu=\ell-1$, that is, $\calM(\calI) =1$ for any instance $\calI$ with $k=2(\ell-1)+1$ districts such that $0$ is the representative of $\ell-1$ districts and $1$ is the representative of $\ell$ districts. 
Using this, we will show that the lemma is also true for $\mu=\ell$, that is, $\calM(\calJ) = 1$ for any instance $\calJ$ with $k=2\ell+1$ districts such that $0$ is the representative of $\ell$ districts and $1$ is the representative of $\ell+1$ districts.

As in the base case, we will define three instances with particular properties. 
The first instance $\calI_1^{(\ell)}$ will have $2\ell-1$ districts partitioned into two sets of identical districts, one of cardinality $\ell$ and one of cardinality $\ell-1$. Using the induction hypothesis, we will argue that the representative of each of the first identical $\ell$ districts in this instance must be $0$. This will then be used to show that $\calM(\calI_2^{(\ell)}) = 0$ for a particularly structured single-district instance $\calI_2^{(\ell)}$. Finally, in the third instance $\calI_3^{(\ell)}$, we will have $\ell$ districts identical to (the district of) $\calI_2^{(\ell)}$ and $\ell+1$ districts, in which $1$ will be the unanimous representative, and will show that it must be $\calM(\calI_3^{(\ell)}) = 1$. 

\medskip
	
\noindent 
\textbf{Instance $\calI_1^{(\ell)}$:} 
There are $2\ell-1$ districts.
\begin{itemize}
\item In each of the first $\ell$ districts, $\frac{\lambda}{2}+ \frac{\lambda}{4\ell}$ agents are positioned at $0$ 
and $\frac{\lambda}{2} - \frac{\lambda}{4\ell}$ agents are positioned at $1/2+\delta$.

\item In each of the remaining $\ell-1$ districts, all $\lambda$ agents are positioned at $0$. 
Since $\calM$ is unanimous, the representative of all these districts is $0$. 
\end{itemize}
We will argue that the representative of the first $\ell$ districts must also be $0$. 
Assume otherwise that the representative of these districts is $1$. 
Then, $\calI_1^{(\ell)}$ is an instance such that $0$ is the representative of $\ell-1$ districts and $1$ is the representative of $\ell$ districts. Hence, by the induction hypothesis, we have that $\calM(\calI_1^{(\ell)})=1$. 
The social cost of the two alternative locations is
$$\SC(0) = \ell \cdot \left(\frac{\lambda}{2} - \frac{\lambda}{4\ell}\right)\cdot \left(\frac{1}{2}+\delta\right) = \frac{2\ell\lambda-\lambda+\delta(4\ell\lambda+2\lambda)}{8}$$
and 
\begin{align*}
\SC(1) &= (\ell-1) \cdot \lambda +\ell \left( \frac{\ell\lambda}{2}+\frac{\lambda}{4\ell} \right) + \ell \cdot \left(\frac{\lambda}{2} - \frac{\lambda}{4\ell}\right)\cdot \left(\frac{1}{2}-\delta\right) \\
&= \frac{14\ell\lambda - 7\lambda - \delta(4\ell\lambda-2\lambda)}{8}.
\end{align*}
By choosing $\delta$ to be sufficiently small, we obtain that $\dist(\calM) \geq 7-\varepsilon$, for any $\varepsilon > 0$, a contradiction. Consequently, it must be the case that the representative of each the first $\ell$ districts is $0$.

\medskip
	
\noindent 
\textbf{Instance $\calI_2^{(\ell)}$:} 
There is a single district with $\frac{\lambda}{2}+ \frac{\lambda}{4\ell}$ agents positioned at $1/2-\delta$ and $\frac{\lambda}{2} - \frac{\lambda}{4\ell}$ agents positioned at $1$. Note that in the district of $\calI_2^{(\ell)}$ and in each of the first $\ell$ identical districts of $\calI_1^{(k)}$, the positions of the agents induce the same ordering over the alternative locations $0$ and $1$, and there are no agents that are indifferent between the two alternatives in any of the two different profiles. Since the representative of the first $\ell$ districts of $\calI_1^{(\ell)}$ is $0$ as we argued above, by \cref{lem:spordinal}, $\calM(\calI_2^{(\ell)}) = 0$.

\medskip	

\noindent 
\textbf{Instance $\calI_3^{(\ell)}$:} 
There are $2\ell+1$ districts.
\begin{itemize}
\item In each of the first $\ell$ districts, the positions of the agents are as in the district of $\calI_2^{(\ell)}$: 
$\frac{\lambda}{2}+ \frac{\lambda}{4\ell}$ agents are positioned at $1/2-\delta$, 
and $\frac{\lambda}{2} - \frac{\lambda}{4\ell}$ agents are positioned at $1$. 
By the discussion above, the representative of these districts is $0$.

\item In each of the remaining $\ell+1$ districts, all $\lambda$ agents are positioned at $1$. 
Since $\calM$ is unanimous, the representative of these districts is $1$. 
\end{itemize}
We have that
\begin{align*}
\SC(0) &= (\ell+1)\cdot \lambda
+ \ell \left(\frac{\lambda}{2}+\frac{\lambda}{4\ell}\right)\left(\frac{1}{2}-\delta\right) 
+ \ell \cdot \left(\frac{\lambda}{2}-\frac{\lambda}{4\ell}\right)\\
&= \frac{14\ell\lambda+7\lambda-\delta (4\ell\lambda+2\lambda)}{8}.
\end{align*}
and 
$$\SC(1) = \ell \cdot \left(\frac{\lambda}{2}+\frac{\lambda}{4\ell}\right) \left( \frac{1}{2}+\delta \right) = \frac{2\ell\lambda + \lambda + \delta(4\ell\lambda+2\lambda)}{8}.$$ 
By choosing $\delta$ to be sufficiently small, if $\calM(\calI_3^{(\ell)})= 0$ then $\dist(\calM) \geq 7-\varepsilon$, for any $\varepsilon > 0$. Therefore, it must be $\calM(\calI_3^{(\ell)}) = 1$, concluding the proof of the induction step.
\end{proof}

We are finally ready to prove our lower bound on the distortion of any strategyproof mechanism in the discrete setting.

\begin{theorem}\label{thm:sp6lower}
In the discrete setting, the distortion of any strategyproof mechanism is at least $7-\varepsilon$, for any $\varepsilon > 0$. 
\end{theorem}

\begin{proof}
Assume towards a contradiction that there exists a strategyproof mechanism $\calM$ with distortion less than $7-\varepsilon$, for any $\varepsilon >0$.  To prove the theorem, we will use three instances $\calJ_1$ (with two districts), $\calJ_2$ (with one district) and $\calJ_3$ (with $2\mu+1$ districts); in all these instances the set of alternatives will be $\calA=\{0,1\}$. In $\calJ_1$, we will argue that the representative of the first district must be $1$, which in turn will imply that the representative of the single district of $\calJ_2$ will have to be $1$. Then, we will argue that it must $\calM(\calJ_3)=1$, which will contradict the assumption that the distortion of the mechanism is less than $7-\varepsilon$, for any $\varepsilon > 0$. Let $\delta > 0$ be an infinitesimally small constant.

\medskip

\noindent 
\textbf{Instance $\calJ_1$:} 
There are two districts.
\begin{itemize}
\item In the first district, $\lambda/2$ agents are positioned at $1/2-\delta$, and $\lambda/2$ agents are positioned at $1$. 

\item In the second district, all $\lambda$ agents are positioned at $1$. 
Since $\calM$ is unanimous, the representative of this district is $1$.
\end{itemize}
We will show that the representative of the first district must be $1$. 
Assume otherwise that it is $0$. 
Then, we have one district with $0$ as its representative and one district with $1$ as its representative. 
Recall that for such instances it is without loss of generality to assume that $\calM$ will select the left-most district representative, that is, $\calM(\calJ_1) = 0$.
We have 
$$\SC(0) = \frac{\lambda}{2}\left(\frac{1}{2}-\delta\right) +\frac{\lambda}{2} + \lambda = \frac{7\lambda}{4}-\frac{\lambda\delta}{2}$$ 
and
$$\SC(1) = \frac{\lambda}{2}\left(\frac{1}{2}+\delta\right) = \frac{\lambda}{4} + \frac{\lambda\delta}{2}.$$ 
By taking $\delta$ to be sufficiently small, we have that $\dist(\calM) \geq 7-\varepsilon$ for any $\varepsilon>0$, a contradiction. 
Therefore, the representative of the first district must be $1$. 

\medskip
	
\noindent 
\textbf{Instance $\calJ_2$:} 
There is a single district, in which $\lambda/2$ of the agents are positioned at $0$ and $\lambda/2$ agents are positioned at $1/2+\delta$. 
Note that in the district of $\calJ_2$ and in the first district of $\calJ_1$, the positions of the agents induce the same ordering over $0$ and $1$, and there are no agents that are indifferent between the two locations in any of the two profiles. 
Since the representative of the first district of $\calJ_1$ is $1$ as argued above, by \cref{lem:spordinal}, 
the representative of the district of $\calJ_2$ must be $1$ as well.

\medskip	

\noindent 
\textbf{Instance $\calJ_3$:} 
There are $2\mu+1$ districts.
\begin{itemize}
\item In each of the first $\mu$ districts, all $\lambda$ agents are positioned at $0$. 
By unanimity, the representative of these districts is $0$.

\item Each of the remaining $\mu+1$ districts is identical to the district of $\calJ_2$:
there are $\lambda/2$ agents positioned at $0$ and $\lambda/2$ agents positioned at $1/2+\delta$. 
By the discussion above, the representative of these districts is $1$.
\end{itemize}
From \cref{lem:011winner}, we have that $\calM(\calJ_3) = 1$. 
By the definition of the districts,
$$\SC(0) = (\mu+1)\cdot\frac{\lambda}{2}\left(\frac{1}{2}+\delta\right)= \frac{\mu\lambda}{4} + \frac{\lambda}{4} + \frac{(\mu+1) \lambda \delta}{2},$$ 
and
$$\SC(1) = (\mu+1)\cdot\frac{\lambda}{2}\left(\frac{1}{2}-\delta\right) + (\mu+1)\cdot \frac{\lambda}{2} + \mu \cdot \lambda 
= \frac{7\mu\lambda}{4}+\frac{3\lambda}{4}-\frac{(\mu+1) \lambda \delta}{4}.$$ 
By choosing $\delta$ to be small enough and $\mu$ to be large enough, we have that $\dist(\calM) \geq 7-\varepsilon$, for any $\varepsilon>0$.

\medskip
	
\noindent This completes the proof.
\end{proof}	

We remark that the lower bound of $7$ also extends to \emph{ordinal mechanisms}, whose decisions are based only on the orderings over the alternative locations induced by the positions of the agents, rather than the exact positions themselves. This follows by observing that the property established in \cref{lem:spordinal} is satisfied trivially by ordinal mechanisms (even without the strategyproofness requirement). Moreover, since $\DDM$ is an ordinal mechanism (observe that to pinpoint the median agent within a district and then her closest alternative, we do not really need to know the exact positions of the agents), this bound is also matched from above, leading to the following corollary.

\begin{corollary}
The distortion of any ordinal mechanism is at least $7-\varepsilon$, for any $\varepsilon > 0$.
Moreover, there exists an ordinal mechanism with distortion at most $7$.
\end{corollary}

%%%%%%
%%%%%%

\section{Mechanisms for the continuous setting}\label{sec:continuous}

So far, we focused on the discrete setting, and showed that $\DMM$ achieves the best-possible distortion of $3$ among all mechanisms, while $\DDM$ achieves the best-possible distortion of $7$ among all strategyproof mechanisms. 
We now turn our attention to the continuous setting. We will first present a strategyproof mechanism with distortion $3$ in \cref{sec:continuous-upper}, followed by a lower bound of $2$ for all mechanisms and a lower bound of $3$ for all strategyproof mechanisms in Sections~\ref{sec:continuous-lower-general} and~\ref{sec:continuous-lower-sp}, respectively. 

\subsection{A strategyproof mechanism with distortion $3$}\label{sec:continuous-upper}

Let us recall how $\DMM$ and $\DDM$ choose the representative of a district in the discrete setting. $\dmm$ chooses the alternative location that minimizes the social cost of the agents whereas $\ddm$ chooses the location that is closer to the median agent. 
In the continuous setting, where the set of alternative locations is $\RR$, the location of the median agent is known to minimize the social cost of the agents in a district, and thus the continuous version of $\ddm$, which chooses as representative the position of the median agent, is an implementation of $\dmm$. As we show with the following theorem, $\CDM$ ($\cdm$) inherits the best properties of $\dmm$ and $\ddm$, leading  to the following statement

\begin{theorem}\label{thm:continuous-mechanism}
$\CDM$ is strategyproof and has distortion at most $3$.
\end{theorem}

Our proof of \cref{thm:continuous-mechanism} relies heavily on the techniques used in the proof of \cref{thm:sp-mechanisms} to show that $\ddm$ is strategyproof, and in the proof of \cref{thm:DMM-distortion} to show that the distortion of $\dmm$ is at most $3$. In particular, the arguments used in those proofs are adapted to the continuous version to accommodate the fact that the set of alternatives is $\RR$. To avoid being repetitive, we give a full proof in the appendix. As already briefly discussed in \cref{sec:related}, the proof of the distortion bound in \cref{thm:continuous-mechanism} also follows from the work of Procaccia and Tennenholtz~\cite{procaccia2009approximate}, who considered a setting with \emph{super-agents} that control multiple locations, and their cost is the total distance between those locations and the facility. They showed that the {\em median-of-medians} mechanism is $3$-approximate. By interpreting the super-agents as district representatives in our case, so that the social cost objectives in the two settings coincide, we obtain the theorem.

\subsection{Lower bounds on the distortion}\label{sec:continuous-lower}

We now turn our attention to showing lower bounds on the distortion of mechanisms in the continuous setting.

\subsubsection{An unconditional lower bound}\label{sec:continuous-lower-general}

We start with a lemma quite similar to \cref{lem:011winner-general}. A key difference is that we no longer have two fixed alternative locations as we did in the discrete setting, so we will establish the lemma for two arbitrary locations $y_1$ and $y_2$ with $y_1 < y_2$. 

\begin{lemma}\label{lem:contabb-general}
Let $\calM$ be any mechanism with distortion strictly less than $2$. 
Let $\calI$ be any instance with $k=2\mu+1$ districts such that 
(a) the representative of $\mu$ districts is some location $y_1$, 
(b) the representative of the remaining $\mu+1$ districts is some location $y_2$, and 
(c) $y_1 < y_2$.
Then,
\begin{itemize}
\item[(i)] $\calM(\calI)=y_2$, and 
\item[(ii)] for the representative $z$ of any district in which 
	\begin{itemize}
	\item[\normalfont \textbullet] $\frac{\lambda}{2} + \frac{\lambda}{4(\mu+1)}$ agents are positioned at some $z_1$,
	\item[\normalfont \textbullet] $\frac{\lambda}{2} - \frac{\lambda}{4(\mu+1)}$ agents are positioned at some $z_2 > z_1$,
	\end{itemize}	
it holds that $z < \frac{z_1+z_2}{2}$.
\end{itemize}
\end{lemma}

\begin{proof}
We will prove the statement by induction on $\mu$. 

\bigskip

\noindent 
\underline{Base case: $\mu=1$.} \\
Let $y_1$ and $y_2 > y_1$ be any real numbers. Consider an instance $\calI$ with the following three districts:
\begin{itemize}
\item In the first district, all $\lambda$ agents are positioned at $y_1$. 
By unanimity, the representative of this district is $y_1$.

\item In the other two districts, all $\lambda$ agents are positioned at $y_2$. 
Again by unanimity, the representative of these districts is $y_2$.
\end{itemize}
Clearly, the social costs of the two representative locations are
$\SC(y_1) = 2\lambda(y_2-y_1)$  and $\SC(y_2) = \lambda(y_2-y_1)$. 
Since $\dist(\calM) <2$, it must be the case that the $\calM(\calI)=y_2$. 

For part (ii), let $z_1$ and $z_2 > z_1$ be any real numbers. We will show that some location $z < \frac{z_1+z_2}{2}$ must be the representative of a district $d$ such that (a) $5\lambda/8$ agents are positioned at $z_1$ and (b) $3\lambda/8$ agents are positioned at $z_2$.
Consider the following instance $\calJ$ with three districts: 
\begin{itemize}
\item The first two districts are identical to district $d$ described above. 
Let $z$ be the representative of these districts.

\item In the third district, all $\lambda$ agents are positioned at $z_1$. 
By unanimity, the representative of this district is $z_1$.
\end{itemize}
Assume towards a contradiction that $z \geq \frac{z_1+z_2}{2}$. 
Then, by part (i) of the statement proved above for the base case (which holds for any $y_1$ and $y_2 > y_1$), we know that $\calM(\calJ)=z$. 
We have
$$\SC(z_1) = 2 \cdot \frac{3\lambda}{8}(z_2-z_1) = \frac{3\lambda}{4}(z_2-z_1)$$
and 
$$\SC(z) = 2\cdot \frac{5\lambda}{8} (z-z_1) + 2\cdot \frac{3\lambda}{8} |z_2-z| + \lambda (z-z_1).$$ 
Observe that $\SC(z)$ is an increasing function of $z$, no matter whether $z < z_2$ or $z \geq z_2$. 
Since $z \geq \frac{z_1+z_2}{2}$,
$$\SC(z) \geq \SC\left(\frac{z_1+z_2}{2}\right) = \frac{6\lambda }{4}(z_2-z_1).$$
Therefore, we have $\dist(\calM) \geq \dist(\calJ|\calM) \geq 2$, a contradiction. See also \cref{fig:lower-continuous-basecase}.
\bigskip

\noindent 
\underline{Induction step:} \\
We will prove the statement for $\mu=\ell$, assuming that it holds for $\mu=\ell-1$. 

Let $y_1$ and $y_2 > y_1$ be any real numbers. Consider the following instance $\calI$ with $2\ell+1$ districts:
\begin{itemize}
\item In each of the first $\ell$ districts, $\frac{\lambda}{2} + \frac{\lambda}{4\ell}$ agents are positioned at $y_1$ and $\frac{\lambda}{2} - \frac{\lambda}{4\ell}$ agents are positioned at $y_2$. 
By part (ii) of the induction hypothesis, the representative of these districts is some location $z \leq \frac{y_1+y_2}{2}$. 

\item In each of the other $\ell+1$ districts, all $\lambda$ agents are positioned at $y_2$. 
By unanimity, the representative of these districts is $y_2$.
\end{itemize}
By the range of possible values of $z$, we have 
\begin{align*}
\SC(z) &= 
\ell\cdot \left(\frac{\lambda}{2}+\frac{\lambda}{4\ell}\right)\cdot |z-y_1| + \ell\cdot \left(\frac{\lambda}{2}-\frac{\lambda}{4\ell}\right)\cdot (y_2-z) + (\ell+1)\cdot \lambda(y_2-z) \\
&\geq \lambda \cdot \frac{(2\ell+1)(y_2-y_1)}{2}
\end{align*}
and 
$$ \SC(y_2) = \ell \left(\frac{\lambda}{2}+\frac{\lambda}{4\ell}\right)(y_2-y_1) = \lambda \cdot \frac{2\ell+1}{4} (y_2-y_1).$$
If $\calM(\calI)=z$, then $\dist(\calM) \geq \dist(\calI | \calM) \geq 2$, a contradiction. 
Consequently, it must be the case that $\calM(\calI)=y_2$. 

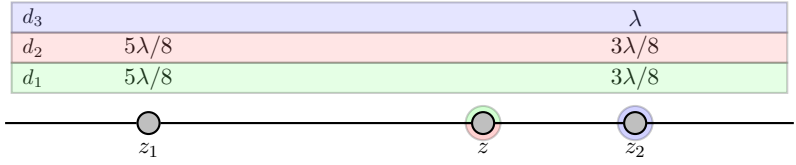
\begin{figure}
	\centering
	\begin{tikzpicture}[thick,scale=0.8, every node/.style={scale=0.8}]
	\node [circle, draw=black,color=black,fill=white!50!gray, label={below:$z_1$}] (0) at (-4, 0) {};
	\node [semicircle, draw=black,fill=white!10!blue,opacity=0.2] (x) at (4, 0.125) {};
	\node [semicircle, draw=black,rotate=180,fill=white!10!blue,opacity=0.2] (x) at (4, -0.125) {};
	\node [circle, draw=black,color=black,fill=white!50!gray, label={below:$z_2$}] (1) at (4, 0) {};
	\node [semicircle, draw=black,fill=white!10!green,opacity=0.2] (x) at (1.5, 0.125) {};
	\node [semicircle, draw=black,rotate=180,fill=white!10!red,opacity=0.2] (x) at (1.5, -0.125) {};
	\node [circle, draw=black,color=black,fill=white!50!gray, label={below:$z$}] (z) at (1.5, 0) {};
	\node (l) at (-6.5,0) {};
	\node (r) at (6.75,0) {};
	\draw (l) to (0);
	\draw (0) to (z);
	\draw (z) to (1);
	\draw (1) to (r);
	\node (district1) at (-4,0.75) {$5\lambda/8$};
	\node (district1) at (4,0.75) {$3\lambda/8$};
	\node at (-5.9,0.75) {\small{$d_1$}};
	\draw[draw=black,fill=white!50!green,opacity=0.2] (-6.25,0.5) rectangle ++(12.75,0.5);
	\node (district2) at (-4,1.25) {$5\lambda/8$};
	\node (district2) at (4,1.25) {$3\lambda/8$};
	\node at (-5.9,1.25) {\small{$d_2$}};
	\draw[draw=black,fill=white!50!red,opacity=0.2] (-6.25,1) rectangle ++(12.75,0.5);
	\node (district3) at (4,1.75) {$\lambda$};
	\node at (-5.9,1.75) {\small{$d_3$}};
	\draw[draw=black,fill=white!50!blue,opacity=0.2] (-6.25,1.5) rectangle ++(12.75,0.5);
	\end{tikzpicture}	
	\caption{The instance $\calJ$ used in part (ii) of the base case of \cref{lem:contabb-general}. 
	By unanimity, the representative of $d_3$ is $z_2$ (shaded blue). 
	We assume towards a contradiction that the representative of the two identical districts $d_1$ (shaded green) and $d_2$ (shaded red) is some point $z > \frac{z_1+z_2}{2}$ (shaded both green and red); in the figure it is shown to be below $z_2$. Then, by comparing the social cost of $z$ against the social cost of $z_1$, we obtain that $\dist(\calJ|\calM)\geq 2$.}
	\label{fig:lower-continuous-basecase}
\end{figure}

For part (ii), let $z_1$ and $z_2 > z_1$ be any real numbers, and consider the following instance $\calJ$ with $2\ell+1$ districts:
\begin{itemize}
\item In the first $\ell$ districts, all $\lambda$ agents are positioned at $z_1$. 
By unanimity, the representative of these districts is $z_1$.

\item In each of the remaining $\ell+1$ districts, $\frac{\lambda}{2}+\frac{\lambda}{4(\ell+1)}$ agents are positioned at $z_1$ and $\frac{\lambda}{2}-\frac{\lambda}{4(\ell+1)}$ agents are located at $z_2$. Let $z$ be the representative of these districts.
\end{itemize}
Assume towards a contradiction (to part (ii) of the lemma) that $z\geq \frac{z_1+z_2}{2}$. 
Then, by the proof of part (i) of the induction step above (which holds for any $y_1$ and $y_2 > y_1$), we know that $\calM(\calJ)=z$. 
By the range of possible values of $z$, we have
$$\SC(z_1) = (\ell+1)\cdot \left( \frac{\lambda}{2}-\frac{\lambda}{4(\ell+1)}\right) (z_2-z_1) = \lambda \cdot \frac{2\ell+1}{4}(z_2-z_1)$$
and
\begin{align*}
\SC(z) 
&= \ell \cdot \lambda (z-z_1) +  (\ell+1) \cdot \left( \frac{\lambda}{2}+\frac{\lambda}{4(\ell+1)}\right) (z-z_1) + (\ell+1) \cdot \left( \frac{\lambda}{2}-\frac{\lambda}{4(\ell+1)}\right)|z_2-z| \\
&\geq \lambda \cdot \frac{(2\ell+1)(z_2-z_1)}{2}.
\end{align*}
Therefore, $\dist(\calM) \geq \dist(\calJ | \calM) \geq 2$, a contradiction.

\bigskip

\noindent This completes the proof of the lemma.
\end{proof}

We are now ready to prove the main theorem.

\begin{theorem}
In the continuous setting, the distortion of any mechanism is at least $2-\varepsilon$, for any $\varepsilon > 0$.
\end{theorem}

\begin{proof}
Assume towards a contradiction that there exists a mechanism $\calM$ with distortion smaller than $2-\varepsilon$, for any $\varepsilon >0$. First, we will prove that the representative $y$ of a district $d$ such that $\lambda/2$ agents are positioned at $0$ and $\lambda/2$ agents are positioned at $1$, must satisfy $y \geq 1/2$, as otherwise $\dist(\calM) \geq 2$.

\medskip

\noindent Assume that $y < 1/2$ and consider the following instance $\calI$ with two districts: 
\begin{itemize}
\item The first district is identical to district $d$ above.
\item In the second district, all $\lambda$ agents are positioned at $1$. 
By unanimity, the representative of this district is $1$.
\end{itemize}
Recall that for any instance such that there are two districts with different representatives it is without loss of generality to assume that the facility location is the left-most representative, that is, $\calM(\calI)=y$ in our case. 
We have 
$$\SC(y) = \frac{\lambda}{2}\cdot y + \frac{\lambda}{2}\cdot |1-y| + \lambda \cdot |1-y| \geq \frac{\lambda(3-2y)}{2}$$
and 
$$\SC(1) = \frac{\lambda}{2}.$$ 
Therefore, $\dist(\calM) \geq \dist(\calI | \calM) = 3-2y \geq 2$, a contradiction.

\medskip

\noindent 
Now consider the following instance $\calJ$ with $2\mu+1$ districts:
\begin{itemize}
\item In each of the first $\mu$ districts, all $\lambda$ agents are positioned at $0$. 
By unanimity, the representative of these districts is location $0$.

\item Each of the $\mu+1$ districts are identical to $d$: $\lambda/2$ agents are positioned at $0$ and $\lambda/2$ agents are positioned at $1$. By the above discussion, the representative of these districts is some location $y \geq 1/2$.
\end{itemize}
By \cref{lem:contabb-general}, it is $\calM(\calJ) = y$.
We have
$$\SC(0) = (\mu+1)\cdot \frac{\lambda}{2}$$
and 
\begin{align*}
\SC(y) 
&= (\mu+1)\cdot\frac{\lambda}{2}y + (\mu+1)\cdot\frac{\lambda}{2}|1-y| + \mu \cdot \lambda y \\
&\geq (\mu+1)\cdot \frac{\lambda}{2} + \mu \cdot \lambda y.
\end{align*}
Hence, 
$$\dist(\calM) \geq \dist(\calJ | \calM) = 1+\frac{2 \mu y}{\mu+1} \geq 1 + \frac{\mu}{\mu+1}.$$
By choosing $\mu$ to be sufficiently large, the distortion of $\calM$ is at least $2-\varepsilon$, for any $\varepsilon > 0$, a contradiction.
 
\medskip

\noindent This completes the proof.
\end{proof}

Even though we have been unable to show a matching unconditional upper bound, we believe that this should be possible. To this end, we conjecture that there exists a mechanism with distortion $2$ for the continuous setting.

\subsubsection{A lower bound for strategyproof mechanisms}\label{sec:continuous-lower-sp}

Recall that the proof of our lower bound on the distortion of strategyproof mechanisms in the discrete setting relies heavily on \cref{lem:spordinal}. However, that lemma is no longer meaningful in the continuous setting, because every alteration of an agent's position immediately results in a new preference ordering over the locations. Instead, we will exploit the following lemma, which states that if we move {\em any subset} of the $\lambda$ agents in a district to the representative location of the district, then strategyproofness dictates that the representative remains the same. The proof of the lemma can be found in the appendix.

\begin{lemma}\label{lem:contsplemma}
Let $\calM$ be a strategyproof mechanism. 
Let $\xx_d$ be a district position profile and let $y$ be the representative of district $d$. 
Let $S \subseteq \calN_d$ be any subset of the agents in $d$. 
Then, $y$ remains the representative of $d$ under the district position profile $\yy_d$ which is obtained from $\xx_d$ be moving all agents in $S$ to $y$, that is, $y_i = y$ for every agent $i \in S$ and $y_i=x_i$ for every $i \in \calN_d \setminus S$. 
\end{lemma}

Our next lemma is very similar to \cref{lem:contabb-general}, 
but applies only to \emph{strategyproof} mechanisms with distortion strictly less than $3$. %-\varepsilon$, for any $\varepsilon > 0$.

\begin{lemma}\label{lem:contabb}
Let $\calM$ be any strategyproof mechanism for the continuous setting with distortion strictly less than $3$. 
Let $\calI$ be any instance with $k=2\mu+1$ districts such that 
(a) the representative of $\mu$ districts is some location $y_1$, 
(b) the representative of the remaining $\mu+1$ districts is some location $y_2$, and 
(c) $y_1 < y_2$. 
Then, $\calM(\calI) = y_2$. 
\end{lemma}

\begin{proof}
We will prove the statement by induction on $\mu$.

\bigskip
	
\noindent 
\underline{Base case: $\mu=1$.} \\
Consider an instance $\calI_1$ with a single district in which the first $3\lambda/4$ agents are positioned at $y_1$, and the remaining $\lambda/4$ agents are positioned at $y_2$. We will argue that for the representative $z$ of the district, we have $z < y_2$; this is obvious when $z\leq y_1$, therefore assume $z > y_1$.
Its social cost is
$$
\SC(z) = \frac{3\lambda}{4}(z-y_1) + \frac{\lambda}{4}|y_2-z| \geq \frac{\lambda}{4}(2z+y_2-3y_1). 
$$
Observe that $y_1$ is the location that minimizes the social cost to $\SC(y_1) = \frac{\lambda}{4}(y_2-y_1)$. 
Since $\dist(\calM) < 3$, it has to be the case that $z < y_2$.

Next, consider an instance $\calI_2$ with a single district such that
the first $3\lambda/4$ agents are positioned at $z$,
and the remaining $\lambda/4$ agents are positioned at $y_2$.
Observe that the districts of $\calI_1$ and $\calI_2$ are the same, with the only difference that the $3\lambda/4$ agents who are positioned at $y_1$ in $\calI_1$ have been moved to $z$ in $\calI_2$. 
Hence, by \cref{lem:contsplemma}, the representative of the district in $\calI_2$ must also be $z$ . 

Finally, consider an instance $\calI_3$ with the following three districts:
\begin{itemize}
\item The first district is identical to the district in $\calI_2$: $3\lambda/4$ agents are positioned at $z$, and  $\lambda/4$ agents are positioned at $y_2$. By the above discussion, the representative of this district is $z < y_2$.

\item In the remaining two districts, all $\lambda$ agents are positioned at $y_2$. 
By unanimity, the representative of these districts is $y_2$.
\end{itemize}
We have
$$\SC(z) =  \frac{\lambda}{4}(y_2-z) +  2\lambda(y_2-z) =\frac{9\lambda}{4}(y_2-z)$$
and
$$\SC(y_2) = \frac{3\lambda}{4}(y_2-z).$$ 
If $\calM(\calI_3)=z$, then $\dist(\calM) \geq \dist(\calI_3 | \calM) \geq 3$. Hence, it must to be the case that $\calM(\calI_3) = y_2$. See also \cref{fig:lower-spcont-basecase}.

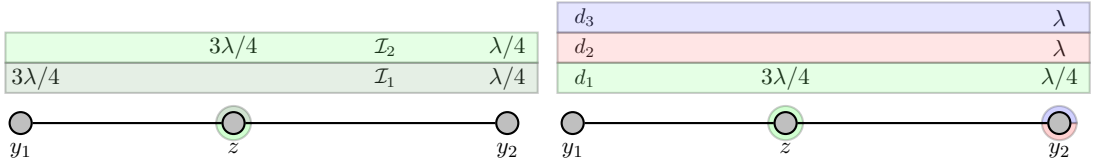
\begin{figure}
	\centering
	\begin{minipage}[b]{0.45\linewidth}
		\centering
		\begin{tikzpicture}[thick,scale=0.8, every node/.style={scale=0.8}]
		\node [circle, draw=black,color=black,fill=white!50!gray, label={below:$y_1$}] (0) at (-4, 0) {};
		\node [circle, draw=black,color=black,fill=white!50!gray, label={below:$y_2$}] (1) at (4, 0) {};
		\node [semicircle, draw=black,fill=white!10!DarkGreen,opacity=0.2] (x) at (-0.5, 0.125) {};
		\node [semicircle, draw=black,rotate=180,fill=white!10!green,opacity=0.2] (x) at (-0.5, -0.125) {};
		\node [circle, draw=black,color=black,fill=white!50!gray, label={below:$z$}] (z) at (-0.5, 0) {};
		\draw (0) to (z);
		\draw (z) to (1);
		\node (district1) at (-3.75,0.75) {$3\lambda/4$};
		\node (district1) at (4,0.75) {$\lambda/4$};
		\node at (2,0.75) {\small{$\calI_1$}};
		\draw[draw=black,fill=white!50!DarkGreen,opacity=0.2] (-4.25,0.5) rectangle ++(8.75,0.5);
		\node (district2) at (-0.5,1.25) {$3\lambda/4$};
		\node (district2) at (4,1.25) {$\lambda/4$};
		\node at (2,1.25) {\small{$\calI_2$}};
		\draw[draw=black,fill=white!50!green,opacity=0.2] (-4.25,1) rectangle ++(8.75,0.5);
		\end{tikzpicture}
	\end{minipage}
	\begin{minipage}[b]{0.45\linewidth}
		\centering
		\begin{tikzpicture}[thick,scale=0.8, every node/.style={scale=0.8}]
		\node [circle, draw=black,color=black,fill=white!50!gray, label={below:$y_1$}] (0) at (-4, 0) {};
		\node [semicircle, draw=black,fill=white!10!blue,opacity=0.2] (x) at (4, 0.125) {};
		\node [semicircle, draw=black,rotate=180,fill=white!10!red,opacity=0.2] (x) at (4, -0.125) {};
		\node [circle, draw=black,color=black,fill=white!50!gray, label={below:$y_2$}] (1) at (4, 0) {};
		\node [semicircle, draw=black,fill=white!10!green,opacity=0.2] (x) at (-0.5, 0.125) {};
		\node [semicircle, draw=black,rotate=180,fill=white!10!green,opacity=0.2] (x) at (-0.5, -0.125) {};
		\node [circle, draw=black,color=black,fill=white!50!gray, label={below:$z$}] (z) at (-0.5, 0) {};
		\draw (0) to (z);
		\draw (z) to (1);
		\node (district1) at (-0.5,0.75) {$3\lambda/4$};
		\node (district1) at (4,0.75) {$\lambda/4$};
		\node at (-3.8,0.75) {\small{$d_1$}};
		\draw[draw=black,fill=white!50!green,opacity=0.2] (-4.25,0.5) rectangle ++(8.75,0.5);
		\node (district2) at (4,1.25) {$\lambda$};
		\node at (-3.8,1.25) {\small{$d_2$}};
		\draw[draw=black,fill=white!50!red,opacity=0.2] (-4.25,1) rectangle ++(8.75,0.5);
		\node (district3) at (4,1.75) {$\lambda$};
		\node at (-3.8,1.75) {\small{$d_3$}};
		\draw[draw=black,fill=white!50!blue,opacity=0.2] (-4.25,1.5) rectangle ++(8.75,0.5);
		\end{tikzpicture}
	\end{minipage}
	\caption{The base case of \cref{lem:contabb}. 
	\emph{Left:} Instance $\calI_1$ consists of a single district with representative $z$ (shaded dark green). 
	By \cref{lem:contsplemma}, $z$ is also the representative of the single district of instance $\calI_2$ (shaded green). 
	\emph{Right:} In instance $\calI_3$, $d_1$ is identical to the single district of instance $\calI_2$, and thus its representative is $z$ (shaded green). By unanimity, the representative of both $d_2$ (shaded red) and $d_3$ (shaded blue) is $y_2$ (shaded both red and blue). 
	If $\calM(\calI_3) = z$, then $\dist(\calI_3|\calM) = 3$, which implies that $\calM(\calI_3)=y_2$.}
	\label{fig:lower-spcont-basecase}
\end{figure}
\bigskip
	
\noindent 
\underline{Induction step:} \\
We will now prove the statement for $\mu=\ell$, assuming that it holds for $\mu=\ell-1$. 

First, consider an instance $I^{(\ell)}$ with the following $2\ell-1$ districts:
\begin{itemize}
\item In each of the first $\ell-1$ districts, all $\lambda$ agents are positioned at some position $y$. 
By unanimity, the representative of these districts is $y$.

\item In each of the remaining $\ell$ districts, $\lambda/2 + \lambda/(4\ell)$ agents are positioned at $y$, and $\lambda/2 - \lambda/(4\ell)$ agents are positioned at some $y_2>y$. Let $z$ be the representative of these districts. 
\end{itemize}
Again, we will argue that $z<y_2$; this is obviously true when $z\leq y$, therefore assume $z>y$. In that case, Observe that it must be $\calM(\calI^{(\ell)}) = z$ by the induction hypothesis. We have
\begin{align*}
\SC(z) 
&= (\ell-1) \cdot \lambda (z-y) + \ell \cdot \left( \frac{\lambda}{2} + \frac{\lambda}{4\ell} \right)(z-y) + \ell \cdot \left( \frac{\lambda}{2} - \frac{\lambda}{4\ell} \right)|y_2-z| \\
&\geq \lambda \frac{2\ell-1}{4}(2z+y_2-3y).
\end{align*}
At the same time, we have that 
$$\SC(y) = \lambda \frac{2\ell-1}{4}(y_2-y).$$
and hence, since $\dist(\calM) < 3$, it must be the case that $z < y_2$. 

Our next goal is to identify a district $d$ such that $\lambda/2 + \lambda/(4\ell)$ agents are positioned at some location $y_1<y_2$, $\lambda/2 - \lambda/4(\ell)$ agents are positioned at $y_2$, and the representative of the district is $y_1$. 
\begin{itemize}
\item If $z=y$, then any of the last $\ell$ districts in $\calI^{(\ell)}$ is such a district.

\item If $z \neq y$, consider a district $d$ such that $\lambda/2 + \lambda/(4\ell)$ agents are positioned at $z$ and $\lambda/2 - \lambda/(4\ell)$ agents are positioned at $y_2$. Observe that this district is similar to each of the last $\ell$ districts in $\calI^{(\ell)}$, with the difference that the $\lambda/2+\lambda/4\ell$ agents who are positioned at $y$ in $\calI^{(\ell)}$ are now moved to $z$. 
Therefore, by \cref{lem:contsplemma}, the representative of $d$ must be $z$, and since $z <y_2$, we have identified the desired instance.
\end{itemize}
So, in any case we have identified the district $d$ we have been looking for, with $y_1=z$.
	
Finally, consider an instance $\calJ^{(\ell)}$ with the following $2\mu+1$ districts:
\begin{itemize}
\item Each of the first $\ell$ districts is identical to $d$ above: $\lambda/2+\lambda/(4\ell)$ agents are positioned at $y_1$ and $\lambda/2 - \lambda/(4\ell)$ agents are positioned at $y_2$. 
So, the representative of these districts is $y_1$.

\item In each of the other $\ell+1$ districts, all $\lambda$ agents are positioned at $y_2$. 
By unanimity, the representative of these districts is $y_2$.
\end{itemize} 
We have
\begin{align*}
\SC(y_1) = \ell \cdot \left( \frac{\lambda}{2} - \frac{\lambda}{4\ell} \right) (y_2-y_1) 
+ (\ell+1) \cdot \lambda (y_2-y_1)
= \frac{3\lambda(2\ell+1)}{4} (y_2-y_1)
\end{align*}
and
$$\SC(y_2) = \ell \cdot \left( \frac{\lambda}{2} + \frac{\lambda}{4\ell} \right) (y_2-y_1) = \frac{\lambda(2\ell+1)}{4} (y_2-y_1).$$
If $\calM(\calJ^{(\ell)}) = y_1$ then $\dist(\calM) \geq \dist(\calJ^{(\ell)} | \calM) \geq 3$, a contradiction. Hence, it must be $\calM(\calJ^{(\ell)}) = y_2$.
\end{proof}

We are now ready to prove the lower bound.

\begin{theorem}
In the continuous setting, the distortion of any strategyproof mechanism is at least $3-\varepsilon$, for any $\varepsilon > 0$.
\end{theorem}

\begin{proof}
Assume towards a contradiction that there exists a strategyproof mechanism $\calM$ with distortion smaller than $3-\varepsilon$, for any $\varepsilon >0$. We start from an instance $\calI_1$ with a single district, in which $\lambda/2$ agents are positioned at $0$, and $\lambda/2$ agents are positioned at $1$. Let $y$ be the representative of the district (and thus the facility location). 
We will argue that it must be $y\geq 1$.
Assume otherwise that $y< 1$, and let $\calI_2$ be an instance with a single district that is obtained from the district of $\calI_1$ by moving the first $\lambda/2$ agents from $0$ to $y$ (the remaining $\lambda/2$ agents are still positioned at $1$). 
Note that if $y=0$, then $\calI_1 \equiv \calI_2$. 
Therefore, by \cref{lem:contsplemma}, the representative of (the district of) $\calI_2$ is $y$. 
Next, consider an instance $\calI_3$ with the following two districts:
\begin{itemize}
\item The first district is identical to the district of $\calI_2$: $\lambda/2$ agents are positioned at $y$, and $\lambda/2$ agents are positioned at $1$. So, the representative of this district is $y$.

\item In the second district, all $\lambda$ agents are positioned at $1$.
By unanimity, the representative of this district is $1$.
\end{itemize}
We have
\begin{align*}
\SC(y) = \frac{\lambda}{2}(1-y) + \lambda (1-y) = \frac{3\lambda}{2}(1-y)
\end{align*}
and
\begin{align*}
\SC(1) = \frac{\lambda}{2}(1-y).
\end{align*}
Recall that it is without loss of generality to assume that $\calM$ selects the left-most representative for any instance with two districts such that their representatives are difference. So, in our case, $\calM(\calI_3) = y$. However, this leads to $\dist(M) \geq \dist(\calI_3 | \calM) \geq 3$, a contradiction. We have now established that the representative of (the district of) $\calI_1$ must be $y \geq 1$. 
	
To complete the proof, consider an instance $\calI$ with the following $2\mu+1$ districts:
\begin{itemize}
\item In each of the first $\mu$ districts, all $\lambda$ agents are positioned at $0$.
By unanimity, the representative of these districts is $0$.

\item Each of the remaining $\mu+1$ districts is identical to the district of $\calI_1$: $\lambda/2$ agents are positioned at $0$ and $\lambda/2$ agents are positioned at $1$. By the above discussion, the representative of these districts is $y \geq 1$.
\end{itemize}
By \cref{lem:contabb}, it is $\calM(\calI) = y$. We have
\begin{align*}
\SC(0) = (\mu+1) \cdot \frac{\lambda}{2}
\end{align*}
and
\begin{align*}
\SC(y) = \mu \lambda y + (\mu+1)\cdot \frac{\lambda}{2}y + (\mu+1)\cdot\frac{\lambda}{2}(y-1) \geq (3\mu+1) \cdot \frac{\lambda}{2}.
\end{align*}
Therefore,
$$\dist(\calM) \geq \dist(\calI | \calM) \geq \frac{3\mu+1}{\mu+1}.$$
By choosing $\mu$ to be sufficiently large, the distortion of $\calM$ is at least $3-\varepsilon$ for any $\varepsilon >0$.
This concludes the proof.
\end{proof}

%%%%%%
%%%%%%

\section{Extensions and open problems}\label{sec:extensions}

\subsection{Asymmetric districts}
Our discussion in the previous sections revolves around the assumption that the districts are symmetric. In general however, the districts might be {\em asymmetric}, where every district $d \in \calD$ might consist of a different number $n_d$ of agents. It is not hard to observe that our mechanisms ($\dmm$ and $\ddm$) can be applied in the asymmetric case as well. In addition, the structure of their worst-case instances defined in Section~\ref{sec:discrete-worst-case} is exactly the same; the proof of the lemma {\em does not} require that $n_d = \lambda$ for every $d \in \calD$. Exploiting this, we can show the following result, which generalizes Theorems~\ref{thm:DMM-distortion}, \ref{thm:DDM-distortion} and~\ref{thm:continuous-mechanism}. 

\begin{theorem}\label{thm:asymmetric}
Let $\alpha = \frac{\max_{d \in \calD} n_d}{\min_{d \in \calD} n_d}$. The distortion of $\dmm$ is at most $3\alpha$ and the distortion of $\ddm$ is at most $7 \alpha$.
\end{theorem}

Unfortunately, our lower bounds are tailor-made for the symmetric case, and thus it is an interesting open problem to extend them to the case of asymmetric districts. As $\dmm$ and $\ddm$ do not take into account the district sizes, it would also be interesting to see whether using this information could lead to mechanisms with improved distortion guarantees (besides the symmetric case).

\subsection{Proxy voting}
Another ingredient of our distributed setting is that the facility location is chosen from the set of district representatives, thus modeling scenarios in which decisions of independent groups are aggregated into a common outcome.  Alternatively, one could assume that the location can be chosen from the set of {\em all} alternative locations, in which case the district representatives are used as {\em proxies} in a district-based election (e.g. see \cite{anshelevich2021representative} and references therein). This captures situations where the alternatives are agents themselves, and the groups select as representatives those alternatives that more closely reflect their collective opinions.
Since the set of district representatives is a subset of the alternative locations, it is straightforward to see that our upper bounds also hold for this proxy model. Our lower bounds in the discrete setting extend as well, since there are only two alternative locations in the instances used in the proofs, and each of them is a representative for at least one district. Hence, our mechanisms are best possible for the proxy model in the discrete setting. 

\begin{corollary}
In the proxy model, the distortion of $\dmm$ is at most $3$ and the distortion of $\ddm$ is at most $7$. Furthermore, in the discrete setting, $\dmm$ and $\ddm$ are the best possible among general and strategyproof mechanisms. 
\end{corollary}

\noindent
In the continuous setting, our lower bounds do not immediately carry over, and it is an intriguing question to identify the exact bound for general and strategyproof mechanisms.

\subsection{Other directions} 
In terms of extending and generalizing our model, there is ample ground for future work. As is typical in the facility location literature, one could consider objectives different than the social cost, such as the {\em maximum cost} or the {\em sum of squares}. Again, the goal would be to show bounds on the distortion, and also design good strategyproof mechanisms. Other possible extensions could include multiple facilities, more general metric spaces, different cost functions, or studying the many different variants of the facility location problem in the distributed setting.

%%%%%%
%%%%%%

\bibliographystyle{plainnat}
\bibliography{references}

\begin{thebibliography}{65}
\providecommand{\natexlab}[1]{#1}
\providecommand{\url}[1]{\texttt{#1}}
\expandafter\ifx\csname urlstyle\endcsname\relax
  \providecommand{\doi}[1]{doi: #1}\else
  \providecommand{\doi}{doi: \begingroup \urlstyle{rm}\Url}\fi

\bibitem[Abramowitz and Anshelevich(2018)]{abramowitz2017utilitarians}
Ben Abramowitz and Elliot Anshelevich.
\newblock Utilitarians without utilities: Maximizing social welfare for graph
  problems using only ordinal preferences.
\newblock In \emph{Proceedings of the 32nd {AAAI} Conference on Artificial
  Intelligence ({AAAI})}, pages 894--901, 2018.

\bibitem[Abramowitz et~al.(2019)Abramowitz, Anshelevich, and
  Zhu]{abramowitz2019awareness}
Ben Abramowitz, Elliot Anshelevich, and Wennan Zhu.
\newblock Awareness of voter passion greatly improves the distortion of metric
  social choice.
\newblock In \emph{Proceedings of the The 15th Conference on Web and Internet
  Economics ({WINE})}, pages 3--16, 2019.

\bibitem[Alon et~al.(2010)Alon, Feldman, Procaccia, and
  Tennenholtz]{alon2009strategyproof}
Noga Alon, Michal Feldman, Ariel~D. Procaccia, and Moshe Tennenholtz.
\newblock Strategyproof approximation of the minimax on setworks.
\newblock \emph{Mathematics of Operations Research}, 35\penalty0 (3):\penalty0
  513--526, 2010.

\bibitem[Amanatidis et~al.(2021)Amanatidis, Birmpas, Filos{-}Ratsikas, and
  Voudouris]{ABFV20}
Georgios Amanatidis, Georgios Birmpas, Aris Filos{-}Ratsikas, and Alexandros~A.
  Voudouris.
\newblock Peeking behind the ordinal curtain: Improving distortion via cardinal
  queries.
\newblock \emph{Artificial Intelligence}, 296:\penalty0 103488, 2021.

\bibitem[Anshelevich and Postl(2017)]{anshelevich2017randomized}
Elliot Anshelevich and John Postl.
\newblock Randomized social choice functions under metric preferences.
\newblock \emph{Journal of Artificial Intelligence Research}, 58:\penalty0
  797--827, 2017.

\bibitem[Anshelevich and Sekar(2016)]{anshelevich2016blind}
Elliot Anshelevich and Shreyas Sekar.
\newblock Blind, greedy, and random: Algorithms for matching and clustering
  using only ordinal information.
\newblock In \emph{Proceedings of the 30th {AAAI} Conference on Artificial
  Intelligence ({AAAI})}, pages 390--396, 2016.

\bibitem[Anshelevich and Zhu(2017)]{anshelevich2017tradeoffs}
Elliot Anshelevich and Wennan Zhu.
\newblock Tradeoffs between information and ordinal approximation for bipartite
  matching.
\newblock In \emph{Proceedings of the 10th International Symposium on
  Algorithmic Game Theory ({SAGT})}, pages 267--279, 2017.

\bibitem[Anshelevich and Zhu(2018)]{anshelevich2018ordinal}
Elliot Anshelevich and Wennan Zhu.
\newblock Ordinal approximation for social choice, matching, and facility
  location problems given candidate positions.
\newblock In \emph{Proceedings of the 14th International Conference on Web and
  Internet Economics ({WINE})}, pages 3--20, 2018.

\bibitem[Anshelevich et~al.(2018)Anshelevich, Bhardwaj, Elkind, Postl, and
  Skowron]{anshelevich2018approximating}
Elliot Anshelevich, Onkar Bhardwaj, Edith Elkind, John Postl, and Piotr
  Skowron.
\newblock Approximating optimal social choice under metric preferences.
\newblock \emph{Artificial Intelligence}, 264:\penalty0 27--51, 2018.

\bibitem[Anshelevich et~al.(2021{\natexlab{a}})Anshelevich, Filos{-}Ratsikas,
  Shah, and Voudouris]{survey}
Elliot Anshelevich, Aris Filos{-}Ratsikas, Nisarg Shah, and Alexandros~A.
  Voudouris.
\newblock Distortion in social choice problems: The first 15 years and beyond.
\newblock \emph{CoRR}, abs/2103.00911, 2021{\natexlab{a}}.

\bibitem[Anshelevich et~al.(2021{\natexlab{b}})Anshelevich, Fitzsimmons, Vaish,
  and Xia]{anshelevich2021representative}
Elliot Anshelevich, Zack Fitzsimmons, Rohit Vaish, and Lirong Xia.
\newblock Representative proxy voting.
\newblock In \emph{Proceedings of the 35th AAAI Conference on Artificial
  Intelligence ({AAAI})}, pages 5086--5093, 2021{\natexlab{b}}.

\bibitem[Babaioff et~al.(2016)Babaioff, Feldman, and
  Tennenholtz]{babaioff2016mechanism}
Moshe Babaioff, Moran Feldman, and Moshe Tennenholtz.
\newblock Mechanism design with strategic mediators.
\newblock \emph{ACM Transactions on Economics and Computation (TEAC)},
  4\penalty0 (2):\penalty0 1--48, 2016.

\bibitem[Benade et~al.(2017)Benade, Nath, Procaccia, and
  Shah]{benade2017preference}
Gerdus Benade, Swaprava Nath, Ariel~D. Procaccia, and Nisarg Shah.
\newblock Preference elicitation for participatory budgeting.
\newblock In \emph{Proceedings of the 31st {AAAI} Conference on Artificial
  Intelligence ({AAAI})}, pages 376--382, 2017.

\bibitem[Bhaskar et~al.(2018)Bhaskar, Dani, and Ghosh]{bhaskar2018truthful}
Umang Bhaskar, Varsha Dani, and Abheek Ghosh.
\newblock Truthful and near-optimal mechanisms for welfare maximization in
  multi-winner elections.
\newblock In \emph{Proceedings of the 32nd {AAAI} Conference on Artificial
  Intelligence ({AAAI})}, pages 925--932, 2018.

\bibitem[Black(1957)]{black1986theory}
Duncan Black.
\newblock \emph{The theory of committees and elections}.
\newblock Kluwer Academic Publishers, 1957.

\bibitem[Bogomolnaia and Laslier(2007)]{bogomolnaia2007euclidean}
Anna Bogomolnaia and Jean-Fran{\c{c}}ois Laslier.
\newblock Euclidean preferences.
\newblock \emph{Journal of Mathematical Economics}, 43\penalty0 (2):\penalty0
  87--98, 2007.

\bibitem[Boutilier et~al.(2015)Boutilier, Caragiannis, Haber, Lu, Procaccia,
  and Sheffet]{boutilier2015optimal}
Craig Boutilier, Ioannis Caragiannis, Simi Haber, Tyler Lu, Ariel~D. Procaccia,
  and Or~Sheffet.
\newblock Optimal social choice functions: A utilitarian view.
\newblock \emph{Artificial Intelligence}, 227:\penalty0 190--213, 2015.

\bibitem[Cai et~al.(2016)Cai, Filos{-}Ratsikas, and Tang]{cai2016facility}
Qingpeng Cai, Aris Filos{-}Ratsikas, and Pingzhong Tang.
\newblock Facility location with minimax envy.
\newblock In \emph{Proceedings of the 25th International Joint Conference on
  Artificial Intelligence ({IJCAI})}, pages 137--143, 2016.

\bibitem[Caragiannis and Procaccia(2011)]{Caragiannis2011embedding}
Ioannis Caragiannis and Ariel~D. Procaccia.
\newblock Voting almost maximizes social welfare despite limited communication.
\newblock \emph{Artificial Intelligence}, 175\penalty0 (9-10):\penalty0
  1655--1671, 2011.

\bibitem[Caragiannis et~al.(2017)Caragiannis, Nath, Procaccia, and
  Shah]{caragiannis2017subset}
Ioannis Caragiannis, Swaprava Nath, Ariel~D. Procaccia, and Nisarg Shah.
\newblock Subset selection via implicit utilitarian voting.
\newblock \emph{Journal of Artificial Intelligence Research}, 58:\penalty0
  123--152, 2017.

\bibitem[Caragiannis et~al.(2018)Caragiannis, Filos{-}Ratsikas, Nath, and
  Voudouris]{caragiannis2018truthful}
Ioannis Caragiannis, Aris Filos{-}Ratsikas, Swaprava Nath, and Alexandros~A.
  Voudouris.
\newblock Truthful mechanisms for ownership transfer with expert advice.
\newblock \emph{CoRR}, abs/1802.01308, 2018.

\bibitem[Chan et~al.(2021)Chan, Filos{-}Ratsikas, Li, Li, and Wang]{FLsurvey}
Hau Chan, Aris Filos{-}Ratsikas, Bo~Li, Minming Li, and Chenhao Wang.
\newblock Mechanism design for facility location problems: {A} survey.
\newblock \emph{CoRR}, abs/2106.03457, 2021.

\bibitem[Cheng et~al.(2011)Cheng, Yu, and Zhang]{cheng2011mechanisms}
Yukun Cheng, Wei Yu, and Guochuan Zhang.
\newblock Mechanisms for obnoxious facility game on a path.
\newblock In \emph{Proceedings of the 5th International Conference on
  Combinatorial Optimization and Applications ({COCOA})}, pages 262--271, 2011.

\bibitem[Cheng et~al.(2013)Cheng, Han, Yu, and Zhang]{cheng2013obnoxious}
Yukun Cheng, Qiaoming Han, Wei Yu, and Guochuan Zhang.
\newblock Obnoxious facility game with a bounded service range.
\newblock In \emph{Proceedings of the 10th International Conference on Theory
  and Applications of Models of Computation ({TAMC})}, pages 272--281, 2013.

\bibitem[Deligkas et~al.(2021)Deligkas, Filos{-}Ratsikas, and Voudouris]{DFV21}
Argyrios Deligkas, Aris Filos{-}Ratsikas, and Alexandros~A. Voudouris.
\newblock Heterogeneous facility location with limited resources.
\newblock \emph{CoRR}, abs/2105.02712, 2021.

\bibitem[Duan et~al.(2019)Duan, Li, Li, and Xu]{duan2019heterogeneous}
Lingjie Duan, Bo~Li, Minming Li, and Xinping Xu.
\newblock Heterogeneous two-facility location games with minimum distance
  requirement.
\newblock In \emph{Proceedings of the 18th International Conference on
  Autonomous Agents and Multiagent Systems ({AAMAS})}, pages 1461--1469, 2019.

\bibitem[Elkind and Faliszewski(2014)]{elkind2014recognizing}
Edith Elkind and Piotr Faliszewski.
\newblock Recognizing 1-euclidean preferences: An alternative approach.
\newblock In \emph{Proceedings of the 7th International Symposium on
  Algorithmic Game Theory ({SAGT})}, pages 146--157. Springer, 2014.

\bibitem[Escoffier et~al.(2011)Escoffier, Gourves, Thang, Pascual, and
  Spanjaard]{escoffier2011strategy}
Bruno Escoffier, Laurent Gourves, Nguyen~Kim Thang, Fanny Pascual, and Olivier
  Spanjaard.
\newblock Strategy-proof mechanisms for facility location games with many
  facilities.
\newblock In \emph{Proceedings of the 2nd International conference on
  Algorithmic Decision Theory ({ADT})}, pages 67--81, 2011.

\bibitem[Fain et~al.(2019)Fain, Goel, Munagala, and Prabhu]{fain2018random}
Brandon Fain, Ashish Goel, Kamesh Munagala, and Nina Prabhu.
\newblock Random dictators with a random referee: Constant sample complexity
  mechanisms for social choice.
\newblock In \emph{Proceedings of the 33rd AAAI Conference on Artificial
  Intelligence ({AAAI})}, 2019.

\bibitem[Feigenbaum et~al.(2017)Feigenbaum, Sethuraman, and
  Ye]{feigenbaum2013approximately}
Itai Feigenbaum, Jay Sethuraman, and Chun Ye.
\newblock Approximately optimal mechanisms for strategyproof facility location:
  Minimizing \emph{L\({}_{\mbox{p}}\)} norm of costs.
\newblock \emph{Mathematics of Operations Research}, 42\penalty0 (2):\penalty0
  434--447, 2017.

\bibitem[Feldman and Wilf(2013)]{feldman2013strategyproof}
Michal Feldman and Yoav Wilf.
\newblock Strategyproof facility location and the least squares objective.
\newblock In \emph{Proceedings of the 14th ACM conference on Electronic
  commerce ({EC})}, pages 873--890, 2013.

\bibitem[Feldman et~al.(2016)Feldman, Fiat, and Golomb]{feldman2016facility}
Michal Feldman, Amos Fiat, and Iddan Golomb.
\newblock On voting and facility location.
\newblock In \emph{Proceedings of the 2016 {ACM} Conference on Economics and
  Computation ({EC})}, pages 269--286, 2016.

\bibitem[Filos-Ratsikas and Miltersen(2014)]{filos2014truthful}
Aris Filos-Ratsikas and Peter~Bro Miltersen.
\newblock Truthful approximations to range voting.
\newblock In \emph{Proceedings of the 10th International Conference on Web and
  Internet Economics ({WINE})}, pages 175--188, 2014.

\bibitem[Filos-Ratsikas et~al.(2014)Filos-Ratsikas, Frederiksen, and
  Zhang]{Aris14}
Aris Filos-Ratsikas, S{\o}ren Kristoffer~Stiil Frederiksen, and Jie Zhang.
\newblock {Social welfare in one-sided matchings: Random priority and beyond}.
\newblock In \emph{Proceedings of the 7th Symposium of Algorithmic Game Theory
  ({SAGT})}, pages 1--12, 2014.

\bibitem[Filos-Ratsikas et~al.(2017)Filos-Ratsikas, Li, Zhang, and
  Zhang]{filos2015facility}
Aris Filos-Ratsikas, Minming Li, Jie Zhang, and Qiang Zhang.
\newblock Facility location with double-peaked preferences.
\newblock \emph{Autonomous Agents and Multi-Agent Systems}, 31\penalty0
  (6):\penalty0 1209--1235, 2017.

\bibitem[Filos-Ratsikas et~al.(2020)Filos-Ratsikas, Micha, and
  Voudouris]{filos2020distortion}
Aris Filos-Ratsikas, Evi Micha, and Alexandros~A. Voudouris.
\newblock The distortion of distributed voting.
\newblock \emph{Artificial Intelligence}, 286:\penalty0 103343, 2020.

\bibitem[Fong et~al.(2018)Fong, Li, Lu, Todo, and Yokoo]{fong2018facility}
Chi Kit~Ken Fong, Minming Li, Pinyan Lu, Taiki Todo, and Makoto Yokoo.
\newblock Facility location games with fractional preferences.
\newblock In \emph{Proceedings of the 32nd AAAI Conference on Artificial
  Intelligence ({AAAI})}, pages 1039--1046, 2018.

\bibitem[Fotakis and Tzamos(2013)]{fotakis2013winner}
Dimitris Fotakis and Christos Tzamos.
\newblock Winner-imposing strategyproof mechanisms for multiple facility
  location games.
\newblock \emph{Theoretical Computer Science}, 472:\penalty0 90--103, 2013.

\bibitem[Fotakis and Tzamos(2014)]{fotakis2013power}
Dimitris Fotakis and Christos Tzamos.
\newblock On the power of deterministic mechanisms for facility location games.
\newblock \emph{{ACM} Transactions on Economics and Computation}, 2\penalty0
  (4):\penalty0 15:1--15:37, 2014.

\bibitem[Fotakis and Tzamos(2016)]{fotakis2013strategyproof}
Dimitris Fotakis and Christos Tzamos.
\newblock Strategyproof facility location for concave cost functions.
\newblock \emph{Algorithmica}, 76\penalty0 (1):\penalty0 143--167, 2016.

\bibitem[Goel et~al.(2017)Goel, Krishnaswamy, and Munagala]{goel2017metric}
Ashish Goel, Anilesh~K. Krishnaswamy, and Kamesh Munagala.
\newblock Metric distortion of social choice rules: Lower bounds and fairness
  properties.
\newblock In \emph{Proceedings of the 2017 ACM Conference on Economics and
  Computation ({EC})}, pages 287--304, 2017.

\bibitem[Goel et~al.(2018)Goel, Hulett, and Krishnaswamy]{goel2018relating}
Ashish Goel, Reyna Hulett, and Anilesh~K. Krishnaswamy.
\newblock Relating metric distortion and fairness of social choice rules.
\newblock In \emph{Proceedings of the 13th Workshop on Economics of Networks
  ({N}et{E}con)}, page 4:1, 2018.

\bibitem[Goel et~al.(2019)Goel, Krishnaswamy, Sakshuwong, and
  Aitamurto]{goel2016knapsack}
Ashish Goel, Anilesh~K. Krishnaswamy, Sukolsak Sakshuwong, and Tanja Aitamurto.
\newblock Knapsack voting for participatory budgeting.
\newblock \emph{{ACM} Transactions on Economics and Computation}, 7\penalty0
  (2):\penalty0 8:1--8:27, 2019.

\bibitem[Gross et~al.(2017)Gross, Anshelevich, and Xia]{gross2017agree}
Stephen Gross, Elliot Anshelevich, and Lirong Xia.
\newblock Vote until two of you agree: Mechanisms with small distortion and
  sample complexity.
\newblock In \emph{Proceedings of the 31st {AAAI} Conference on Artificial
  Intelligence ({AAAI})}, pages 544--550, 2017.

\bibitem[Hotelling(1990)]{hotelling1990stability}
Harold Hotelling.
\newblock Stability in competition.
\newblock In \emph{The Collected Economics Articles of Harold Hotelling}, pages
  50--63. Springer, 1990.

\bibitem[Kempe(2020)]{kempe2019analysis}
David Kempe.
\newblock An analysis framework for metric voting based on {LP} duality.
\newblock In \emph{Proceedings of the 34th {AAAI} Conference on Artificial
  Intelligence ({AAAI})}, pages 2079--2086, 2020.

\bibitem[Lu et~al.(2009)Lu, Wang, and Zhou]{lu2009tighter}
Pinyan Lu, Yajun Wang, and Yuan Zhou.
\newblock Tighter bounds for facility games.
\newblock In \emph{Proceedings of the 5th International Workshop on Internet
  and Network Economics ({WINE})}, pages 137--148. 2009.

\bibitem[Lu et~al.(2010)Lu, Sun, Wang, and Zhu]{lu2010asymptotically}
Pinyan Lu, Xiaorui Sun, Yajun Wang, and Zeyuan~Allen Zhu.
\newblock Asymptotically optimal strategy-proof mechanisms for two-facility
  games.
\newblock In \emph{Proceedings of the 11th ACM conference on Electronic
  commerce ({EC})}, pages 315--324, 2010.

\bibitem[Mandal et~al.(2019)Mandal, Procaccia, Shah, and
  Woodruff]{mandalefficient}
Debmalya Mandal, Ariel~D. Procaccia, Nisarg Shah, and David~P. Woodruff.
\newblock Efficient and thrifty voting by any means necessary.
\newblock In \emph{Proceedings of the 33rd Conference on Neural Information
  Processing Systems ({NeurIPS})}, pages 7178--7189, 2019.

\bibitem[Mandal et~al.(2020)Mandal, Shah, and Woodruff]{mandalec20}
Debmalya Mandal, Nisarg Shah, and David~P. Woodruff.
\newblock Optimal communication-distortion tradeoff in voting.
\newblock In \emph{Proceedings of the 21st ACM Conference on Economics and
  Computation ({EC})}, 2020.

\bibitem[Moulin(1980)]{moulin1980strategy}
Herv{\'e} Moulin.
\newblock On strategy-proofness and single peakedness.
\newblock \emph{Public Choice}, 35\penalty0 (4):\penalty0 437--455, 1980.

\bibitem[Munagala and Wang(2019)]{munagala2019improved}
Kamesh Munagala and Kangning Wang.
\newblock Improved metric distortion for deterministic social choice rules.
\newblock In \emph{Proceedings of the 2019 {ACM} Conference on Economics and
  Computation ({EC})}, pages 245--262, 2019.

\bibitem[Nisan and Ronen(2001)]{nisan2001algorithmic}
Noam Nisan and Amir Ronen.
\newblock Algorithmic mechanism design.
\newblock \emph{Games and Economic behavior}, 35\penalty0 (1-2):\penalty0
  166--196, 2001.

\bibitem[Peters(2017)]{peters2017recognising}
Dominik Peters.
\newblock Recognising multidimensional euclidean preferences.
\newblock In \emph{Proceedings of the 31st AAAI Conference on Artificial
  Intelligence ({AAAI})}, pages 642--648, 2017.

\bibitem[Peters et~al.(1992)Peters, van~der Stel, and
  Storcken]{peters1992pareto}
Hans Peters, Hans van~der Stel, and Ton Storcken.
\newblock Pareto optimality, anonymity, and strategy-proofness in location
  problems.
\newblock \emph{International Journal of Game Theory}, 21\penalty0
  (3):\penalty0 221--235, 1992.

\bibitem[Procaccia and Rosenschein(2006)]{procaccia2006distortion}
Ariel~D. Procaccia and Jeffrey~S. Rosenschein.
\newblock The distortion of cardinal preferences in voting.
\newblock In \emph{Proceedings of the 10th International Workshop on
  Cooperative Information Agents ({CIA})}, pages 317--331, 2006.

\bibitem[Procaccia and Tennenholtz(2013)]{procaccia2009approximate}
Ariel~D. Procaccia and Moshe Tennenholtz.
\newblock Approximate mechanism design without money.
\newblock \emph{{ACM} Transactions on Economics and Computation}, 1\penalty0
  (4):\penalty0 18:1--18:26, 2013.

\bibitem[Rader(1963)]{rader1963existence}
Trout Rader.
\newblock The existence of a utility function to represent preferences.
\newblock \emph{The Review of Economic Studies}, 30\penalty0 (3):\penalty0
  229--232, 1963.

\bibitem[Schummer and Vohra(2002)]{schummer2002strategy}
James Schummer and Rakesh~V. Vohra.
\newblock Strategy-proof location on a network.
\newblock \emph{Journal of Economic Theory}, 104\penalty0 (2):\penalty0
  405--428, 2002.

\bibitem[Serafino and Ventre(2015)]{serafino2015truthful}
Paolo Serafino and Carmine Ventre.
\newblock Truthful mechanisms without money for non-utilitarian heterogeneous
  facility location.
\newblock In \emph{Proceedings of the 29th AAAI Conference on Artificial
  Intelligence ({AAAI})}, pages 1029--1035, 2015.

\bibitem[Serafino and Ventre(2016)]{serafino2016heterogeneous}
Paolo Serafino and Carmine Ventre.
\newblock Heterogeneous facility location without money.
\newblock \emph{Theoretical Computer Science}, 636:\penalty0 27--46, 2016.

\bibitem[Shmoys et~al.(1997)Shmoys, Tardos, and
  Aardal]{shmoys1997approximation}
David~B. Shmoys, {\'E}va Tardos, and Karen Aardal.
\newblock Approximation algorithms for facility location problems.
\newblock In \emph{Proceedings of the 29th Annual ACM Symposium on Theory of
  Computing ({STOC})}, pages 265--274, 1997.

\bibitem[Sui and Boutilier(2015)]{sui2015approximately}
Xin Sui and Craig Boutilier.
\newblock Approximately strategy-proof mechanisms for (constrained) facility
  location.
\newblock In \emph{Proceedings of the 14th International Conference on
  Autonomous Agents and Multiagent Systems ({AAMAS})}, pages 605--613, 2015.

\bibitem[Sui et~al.(2013)Sui, Boutilier, and Sandholm]{sui2013analysis}
Xin Sui, Craig Boutilier, and Tuomas Sandholm.
\newblock Analysis and optimization of multi-dimensional percentile mechanisms.
\newblock In \emph{Proceedings of the 23rd International Joint Conference on
  Artificial Intelligence ({IJCAI})}, pages 367--374, 2013.

\bibitem[Tang et~al.(2020)Tang, Yu, and Zhao]{tang2020characterization}
Pingzhong Tang, Dingli Yu, and Shengyu Zhao.
\newblock Characterization of group-strategyproof mechanisms for facility
  location in strictly convex space.
\newblock In \emph{Proceedings of the 21st ACM Conference on Economics and
  Computation ({EC})}, 2020.

\end{thebibliography}

%%%%%%%%%%%%%
%%%%%%%%%%%%%

\appendix

\section{Omitted proofs}

\subsection*{Proof of \cref{lem:unanimous}}
Let $\calM$ be a distributed mechanism with finite distortion, and assume towards a contradiction that $\calM$ is {\em not} unanimous. 
This means that there exists an instance $\calI = (\xx,\calZ,\calD)$ such that all agents in some district $d \in \calD$ have the same position $z \in \calZ$, but $z_d \neq z$. By the definition of a mechanism (and in particular its the locality property), the same must be true for the instance $\calJ$ consisting of only district $d$ (and the same agent positions). In this case the social cost of the optimal position $z$ is $0$, whereas the social cost of the location $z_d$ chosen by the mechanism is strictly positive, resulting in an infinite distortion. 
\hfill $\qed$

\subsection*{Proof of \cref{lem:spiout}}
Assume towards a contradiction that $\calM$ is not strategyproof within districts. 
This means that there exists an instance $\calI=(\xx,\calD,\calZ)$ such that for some district $d \in \calD$, some agent $i \in \calN_d$ has a beneficial manipulation over the decision of the representative of the district. 
In particular, agent $i$ can report a position $\tilde{x}_i^{d}$ such that
\begin{itemize}
\item the representative of $d$ when the district position profile is  $\xx_d$ (corresponding to instance $\calI$) is some alternative $z \in \calZ$;

\item the representative of $d$ when the district position profile is  $(\tilde{x}_i^d, \xx_{-i,d})$ (corresponding to instance $\tilde{\calI}$, where the positions of all agents remain unchanged, except the position of agent $i$) is some alternative $y \in \calZ$; 

\item $x_i$ is closer to $y$ than to $z$. 
\end{itemize}
By the locality property of $\calM$, such a manipulation would be possible for any district that is identical to $d$. 
So, we can without loss of generality assume that $\calI$ consists of $k$ identical districts. 
Since $z$ is the representative of all districts in $\calI$, it must be the case that $\calM(\calI)=z$. 
Furthermore, since $\calI$ and $\tilde{\calI}$ differ only on the reported position of agent $i$ (that is, $x_i$ versus $\tilde{x}_i^d$) in one of the districts, the fact that $\calM$ is strategyproof implies that $\calM(\tilde{\calI}) = z$.
By repeating the above argument for every district, we build a sequence of instances starting from $\calI$, such that every two consecutive instances differ only on the reported position of a single agent who instead of $x_i$ reports $\tilde{x}_i^d$. Hence, the representative of the district containing this agent changes from $z$ to $y$, but the facility location chosen by $\calM$ remains $z$. 
Now, consider the last instance $\calJ$ in this sequence, which consists of $k$ districts that are identical to $d$ in $\tilde{\calI}$, and thus have $y$ as their representative. As a result, it must be $\calM(\calJ)=y$, which however contradicts the property $\calM(\calJ)=z$ that $\calJ$ inherits as an instance of the  sequence. 
\hfill $\qed$

%%%%%%%%%%%
%%%%%%%%%%%

\subsection*{Proof of \cref{thm:sp-mechanisms}}
For $\dmm$, it suffices to show that an agent can manipulate the mechanism within some district. 
Consider an instance with two alternative locations $z_0=0$ and $z_1=1$, respectively. The district consists of just two agents located at $x+\varepsilon$ and $1-x$, for some $x < 1/2$. If the agents report their positions truthfully, then the social cost of $z_0$ is $1+\varepsilon$, while that of $z_1$ is $1-\varepsilon$, and thus $\dmm$ will choose $z_1$ as the facility location. However, since the first agent prefers $z_0$, she can misreport her position as $0$. Then, the social cost of $z_0$ is $1-x$, while the social cost of $b$ is $1+x$, leading $\dmm$ to choose $z_0$.    

For $\ddm$, consider any instance $\calI = (\xx,\calD,\calA)$ and let $w=\ddm(\calI)$ be the facility location chosen by the mechanism. Let $i$ be any agent who belongs to some district $d \in \calD$. We will argue that if $i$ misreports her position as being $x$ instead of $x_i$, the distance between $x_i$ and the location $\ddm( (x,\xx_{-i}), \calD,\calA)$ chosen by the mechanism  will be at least the distance between $x_i$ and $w$. We distinguish between two cases:
\begin{itemize}
\item \underline{$i$ is not the median agent of $d$.} In order for $i$ to affect the outcome of the mechanism, she must first become the median of $d$. Let $j$ be the median agent of $d$ and assume that $x_i \leq x_j$; the case $x_i > x_j$ is similar. To become the median, agent $i$ has to report a position $x > x_j$. Then, either the outcome does not change, in which case the agent does not gain anything, or the median among the representatives changes from $w$ to some other location $z \in \calA$ which is the closest to $x$. However, this can only happen when $z > w \geq x_j$, meaning that the distance between $x_i$ and the facility location has increased. 

\item \underline{$i$ is the median agent of $d$.} If $z_d=w$, then $i$ has no incentive to misreport her true position, so assume $z_d < w$; the case $z_d > w$ is symmetric. Since $w$ is the median among all representatives, to affect the outcome of the mechanism, agent $i$ has to deviate to a position $x$ so that the closest alternative $z$ becomes the new median representative. Since this can only happen if $z > w$, the agent has nothing to gain by doing so.
\end{itemize} 
Hence, $\ddm$ is strategyproof.
\hfill $\qed$

\subsection*{Proof of \cref{lem:sc-peaked}}
Let $p \in (z,y)$ by any position. We partition the set of agents $S$ into the following four sets: $A = \{i: x_i \leq z \}$; $B = \{i: x_i \in (z,w] \}$; $\Gamma = \{i: x_i \in (w,y]\}$; $\Delta = \{i: x_i > y\}$. Now, observe that
\begin{itemize}
\item For every $i \in A$, $\delta(x_i,y) = \delta(x_i,z) + \delta(z,y)$ and $\delta(x_i,p) = \delta(x_i,z) + \delta(z,p)$;
\item For every $i \in B$, $\delta(x_i,y) = \delta(z,y)-\delta(x_i,z)$ and $\delta(x_i,p) =  \delta(z,p)-\delta(x_i,z)$;
\item For every $i \in \Gamma$, $\delta(x_i,y) = \delta(z,y)-\delta(x_i,z)$ and $\delta(x_i,p) = \delta(x_i,z) - \delta(z,p)$;
\item For every $i \in \Delta$, $\delta(x_i,y) = \delta(x_i,z) - \delta(z,y)$ and $\delta(x_i,p) = \delta(x_i,z) - \delta(z,p)$.
\end{itemize}
By using $\SC_T(z) = \sum_{i \in T} \delta(x_i,z)$ to denote the total distance of the agents in set $T \in \{S,A,B,\Gamma,\Delta\}$ from $z$, 
and using the above observations, we have
\begin{align*}
\SC_S(z | \xx) 
&= \SC_A(z) + \SC_B(z) + \SC_\Gamma(z) + \SC_\Delta(z), \\
\SC_S(p | \xx) 
&= \SC_A(z) - \SC_B(z) + \SC_\Gamma(z) + \SC_\Delta(z) +  ( |A| + |B| - |\Gamma| - |\Delta|) \delta(z,p), \\
\SC_S(y | \xx) 
&= \SC_A(z) - \SC_B(z) - \SC_\Gamma(z) + \SC_\Delta(z) +  ( |A| + |B| + |\Gamma| - |\Delta|) \delta(z,y). 
\end{align*}
Now, the assumption that $\SC_S(z | \xx) \leq \SC_S(y | \xx)$ implies that
\begin{align}\label{sc-peaked-assumption}
2\SC_B(z) + 2\SC_\Gamma(z) \leq ( |A| + |B| + |\Gamma| - |\Delta|) \delta(z,y).
\end{align}
On the other hand, we want to show that  $\SC(p | \xx) \leq \SC(y | \xx)$, or, equivalently, 
\begin{align*}
2\SC_\Gamma(z) &\leq (|A| + |B| + |\Gamma| - |\Delta|) \delta(z,y) + ( |\Gamma| - |A| - |B| + |\Delta|) \delta(z,p).
\end{align*}
Because of \eqref{sc-peaked-assumption}, the above expression (and thus our goal) is true in case $|\Gamma| \leq |A| + |B| - |\Delta|$. Otherwise, by rearranging terms in \eqref{sc-peaked-assumption}, it becomes
\begin{align*}
\SC_\Gamma(z) \leq ( |A| + |B| + |\Gamma| - |\Delta|) \delta(z,y) - 2\SC_B(z)  - \SC_\Gamma(z),
\end{align*}
and substituting it in the definition of $\SC(p | \xx)$, we obtain
\begin{align*}
\SC_S(p | \xx) 
&= \SC_A(z) - \SC_B(z) + \SC_\Gamma(z) + \SC_\Delta(z) +  ( |A| + |B| - |\Gamma| - |\Delta|) \delta(z,p) \\
&\leq \SC_A(z) - \SC_B(z) - \SC_\Gamma(z) + \SC_\Delta(z) + ( |A| + |B| + |\Gamma| - |\Delta|) \delta(z,y) \\
&\quad - 2\SC_B(z) +  ( |A| + |B| - |\Gamma| - |\Delta|) \delta(z,w) \\
&= \SC_S(y | \xx) - 2\SC_B(z) +  ( |A| + |B| - |\Gamma| - |\Delta|) \delta(z,w) \\
&\leq  \SC_S(y | \xx),
\end{align*}
where the second equality follows by the definition of $\SC(y | \xx)$ above, and the last inequality follows by the assumption that 
$|A| + |B| - |\Gamma| - |\Delta| < 0$.
\hfill $\qed$

\subsection*{Proof of \cref{lem:optimality-preservation}}
Since the two moves are symmetric, it suffices to show the lemma for the case where an agent $i$ with $x_i < p \leq z$ is moved to $p$. Consider any alternative location $y \in \calA$. By the optimallity of $z$ under $\xx$, we have that
\begin{align}
& \SC(z | \xx) \leq \SC(y | \xx) \nonumber \\
&\Leftrightarrow \delta(x_i,z) + \sum_{j \neq i} \delta(x_i,z) \leq \delta(x_i,y) + \sum_{j \neq i} \delta(x_i,p) \nonumber \\
&\Leftrightarrow \delta(x_i,z) - \delta(x_i,y)  \leq \sum_{j \neq i} \delta(x_i,p) - \sum_{j \neq i} \delta(x_i,z). \label{eq:optimality-condition}
\end{align}
If we show that
\begin{align*}
\delta(p,z) - \delta(p,y) \leq \delta(x_i,z) -  \delta(x_i,y),
\end{align*}
then, by \eqref{eq:optimality-condition}, we will obtain
\begin{align*}
&\delta(p,z) - \delta(p,y) \leq \sum_{j \neq i} \delta(x_i,p) - \sum_{j \neq i} \delta(x_i,z) \\
&\Leftrightarrow \delta(p,z) + \sum_{j \neq i} \delta(x_i,z) \leq \delta(p,y) + \sum_{j \neq i} \delta(x_i,p) \\
&\Leftrightarrow \SC(z | \pp) \leq \SC(y | \pp).
\end{align*}

\noindent
Now, let $\Delta_{iz} = \delta(x_i,z) - \delta(p,z) = \delta(x_i,p) > 0$ and $\Delta_{iy} = \delta(x_i,y) - \delta(p,y)$. We will show that $\Delta_{iz} \geq \Delta_{iy}$, which is equivalent to the desired inequality. 
\begin{itemize}
\item 
If $y \leq x_i$, we obviously have that $\Delta_{iy} < 0 < \Delta_{iz}$. 

\item
If $y \in (x_i, p)$, we have $\Delta_{iz} = \delta(x_i,p) = \delta(x_i,y) + \delta(y,p) > \delta(x_i,y) - \delta(p,y) = \Delta_{iy} $.

\item 
If $y \geq p$, the decrease is exactly the same: $\Delta_{iy} = \delta(x_i,p) = \Delta_{iz}$. 
\end{itemize}
This completes the proof.
\hfill $\qed$

%%%%%%%%%%%
%%%%%%%%%%%

\subsection*{Proof of \cref{lem:spordinal}}
Since we focus only on one district, let us enumerate the agents therein as $\calN_d = [\lambda] = \{1, ..., \lambda\}$.
We consider a sequence of district position profiles $\{ \xx_d^{(0)}, \xx_d^{(1)}, ..., \xx_d^{(\lambda)} \}$ for district $d$ such that
\begin{itemize}
\item $\xx_d^{(0)} = \xx_d$; 
\item $\xx_d^{(j)} = (y_j, \xx_{-j,d}^{(j-1)})$, for  $j \in [\lambda-1]$; 
\item $\xx_d^{(\lambda)} = \yy_d$.
\end{itemize}
That is, profile $\xx_d^{(j)}$ is obtained from profile $\xx_d^{(j-1)}$ by changing only the position of agent $j$ from $x_j$ to $y_j$ (even if $x_j = y_j$). 
	
Assume towards a contradiction that $\calM$ outputs a different representative under $\xx_d$ and $\yy_d$; let $\alpha$ and $\beta$ be those locations, respectively. This means that there exists $j \in [\lambda]$ such that the representative under profile $\xx_d^{(j-1)}$ is $\alpha$, while the representative under profile $\xx_d^{(j)}$ is $\beta$. Consider the corresponding agent $j$ whose reported positions $x_j$ and $y_j$ differentiate the two profiles $\xx_d^{(j-1)}$ and $\xx_d^{(j)}$. Note that since both profiles induce the same ordering of alternatives, agent $j$ prefers the same location in both profiles over the other. 
\begin{itemize}
\item If agent $j$ prefers $\beta$ over $\alpha$, then assume that the true district profile is $\xx_d^{(j-1)}$. 
By reporting $y_j$ instead of $x_j$, agent $j$ can manipulate $\calM$ to choose $\beta$ as the representative of $d$ instead of $\alpha$.

\item If agent $j$ prefers $\alpha$ over $\beta$, then assume that the true district profile is $\xx_d^{(j)}$.
By reporting $x_j$ instead of $y_j$, agent $j$ can manipulate $\calM$ to choose $\alpha$ as the representative of $d$ instead of $\beta$.
\end{itemize}
Therefore, in either case, agent $j$ is able to manipulate $\calM$ within $d$, which contradicts the fact that $\calM$ is strategyproof within districts by \cref{lem:spiout}.
\hfill $\qed$

\subsection*{Proof of \cref{thm:continuous-mechanism}}
We follow the steps used in \cref{sec:discrete}. We first establish that $\cdm$ is strategyproof in \cref{thm:CDM-sp} using arguments similar to those in the proof of \cref{thm:sp-mechanisms}. Then, similarly to \cref{lem:discrete-structure}, in \cref{lem:continuous-structure}, we characterize the worst-case instances of the mechanism in terms of distortion by showing that every instance can be transformed into another one satisfying two particular properties. Finally, exploiting our characterization we show the upper bound of $3$ on the distortion of the mechanism in \cref{thm:CDM-distortion}, following the roadmap in the proof of \cref{thm:DMM-distortion}.

\begin{theorem}\label{thm:CDM-sp}
$\cdm$ is strategyproof.
\end{theorem}

\begin{proof}
Consider any instance $\calI = (\xx,\calD)$, and let $w=\cdm(\calI)$.
Let $i$ be any agent in some district $d \in \calD$, and denote by $w_x = \cdm( (x,\xx_{-i}), \calD)$ the facility location chosen by the mechanism when $i$ unilaterally misreports her position as $x \neq x_i$. We will argue that $\delta(x_i,w) \leq \delta(x,w_x)$ for any $x \in \RR$.  We distinguish between two cases:
\begin{itemize}
\item \underline{$i$ is not the median agent of $d$.} In order for $i$ to affect the outcome, she must become the median of $d$. 
Let $j$ be the median agent of $d$ and assume that $x_i \leq x_j$ (the case $x_i > x_j$ is symmetric). 
To become the median, agent $i$ has to report a position $x > x_j$, which is going to be the new representative of $d$.  
Then, either $w_x = w$, or the median among the representatives becomes $w_x = x$. 
However, the latter can only happen when $x > w \geq x_j \geq x_i$ , and thus we overall have that $\delta(x_i,w) \leq \delta(x,w_x)$.

\item \underline{$i$ is the median agent of $d$.} If $w=x_i$, then $i$ has no incentive to misreport her true position, so let us assume that $z_d = x_i < w$ (the case $x_i > w$ is symmetric). Since $w$ is the median among all representatives, to affect the outcome of the mechanism, agent $i$ has to deviate to a position $x$ so that $w_x = x$. Since this can only happen if $x > w$, it will then be $\delta(x_i,w) < \delta(x,w_x)$.
\end{itemize} 
Hence, $\cdm$ is strategyproof.
\end{proof}

Next, we focus on bounding the distortion of the mechanism. Similarly to the analysis of $\dmm$ and $\ddm$ in the discrete setting, we first  characterize the structure of worst-case instances for $\cdm$. It turns out that the worst-case instance for $\cdm$ have the exact same structure as those for $\dmm$ and $\ddm$. In particular, let $\wc(\cdm)$ be the class of instances $\calI=(\xx,\calD)$ such that
\begin{itemize}
\item[(P1)]
For every agent $i \in \calN$, 
\begin{itemize}
\item $x_i \geq \cdm(\calI)$ if $\cdm(\calI) < \opt(\calI)$, or 
\item $x_i \leq \cdm(\calI)$ if $\cdm(\calI) > \opt(\calI)$.
\end{itemize}

\item[(P2)]
For every $z \in \RR$ which is representative for a set of districts $\calD_z \neq \varnothing$,  
the positions of all agents in the districts of $\calD_z$ are in the interval defined by $z$ and $\opt(\calI)$.
\end{itemize}
We now have the following characterization lemma.

\begin{lemma}\label{lem:continuous-structure}
The distortion of $\cdm$ is equal to 
$$\sup_{\calI \in \wc(\cdm)} \dist(\calI | \cdm).$$
\end{lemma}

\begin{proof}
We follow the reasoning used in the proof of \cref{lem:discrete-structure} for the mechanisms in the discrete setting. 
We transform every instance $\calJ \not\in \wc(\cdm)$ with $\calM(\calJ) = w < o = \opt(\calJ)$ to an instance $\calI \in \wc(\cdm)$ as follows:
\begin{itemize}
\item[(T1)] We move every agent with position strictly smaller than $w$ to $w$.

\item[(T2)] For every location $z$ which is representative for a non-empty set of districts in $\calJ$, we move every agent therein whose position does not lie in the interval defined by $z$ and $o$ to the boundaries of this interval.
\end{itemize}
We will argue that the sequence of instances obtained by the above transformations satisfy the following three properties, which by induction imply that $\dist(\calJ | \cdm) \leq \dist(\calI | \cdm)$:
\begin{itemize}
\item The facility location chosen by the mechanism is always $w$;
\item The optimal location is always $o$;
\item For any two consecutive intermediate instances with position profiles $\xx$ and $\yy$, $\frac{\SC(w|\xx)}{\SC(o|\xx)} \leq \frac{\SC(w|\yy)}{\SC(o|\yy)}$.
\end{itemize}
The proofs of the second and third properties are similar to the proofs of the corresponding properties in \cref{lem:discrete-structure} with the only difference that the set of alternatives is $\RR$ instead of $\calA$. 
Furthermore, the proof of the first property resembles the proof of the corresponding property for $\ddm$ in \cref{lem:discrete-structure}. For completeness, since the aggregation step of $\cdm$ in the districts is a bit different than that of $\ddm$, we present a self-contained proof for the first property, which is overall much simpler.

For (T1), consider any instance in the sequence with position profile $\xx$ such that there exists a district $d \in \calD$ with representative $z_d=z$, which contains some agent $i$ with position $x_i < w$ who is moved to $w$. 
We distinguish between two cases:
\begin{itemize}
\item $z > w$. Clearly, agent $i$ is not the median in $d$, and thus moving her to $w$ will not change the representative of $d$ nor the outcome of the mechanism.

\item $z \leq w$. By moving agent $i$ to $w$, the representative of $d$ can change from $z$ to $w$ if $i$ becomes the median agent. However, the location chosen by $\ddm$ will remain the same, since $w$ will remain the median representative. 
\end{itemize}
For (T2), consider any instance such that there exists an alternative location $z$ which is representative for a non-empty set of districts $\calD_z$, and a district $d \in \calD_z$ contains  an agent $i$ with position $x_i$ outside the interval defined by $z$ and $o$. Since $x_i \neq z$, $i$ clearly cannot be the median agent of $d$, and thus moving $i$ to either $z$ or to $o$ will not change the representative of $d$ nor the outcome of $\cdm$.
\end{proof}

Given the above characterization lemma about the worst-case instances, we are now ready to complete the proof of the theorem by bounding the distortion of $\cdm$. Similarly to the notation used in \cref{sec:discrete-distortion}, let $\calD_z$ be the set of districts for which $z \in \RR$ is the representative, let $Z = \{z \in \RR: \calD_z \neq \varnothing \}$ be the set of all alternative locations which are representative for at least one district, and for every $z \in Z$, $y \in \RR$ let 
\begin{align*}
\SC_z(y | \xx) = \sum_{d \in \calD_z} \sum_{i \in \calN_d} \delta(x_i, y)
\end{align*}
be the total distance of all the agents in the districts of $\calD_z$ from $y$. Also, recall that each district contains exactly $\lambda$ agents. The arguments used in the proof of our next statement follow closely those used in the proof of \cref{thm:DMM-distortion}.

\begin{theorem}\label{thm:CDM-distortion}
The distortion of $\cdm$ is at most $3$.
\end{theorem}

\begin{proof}
Consider any instance $\calI = (\xx,\calD) \in \wc(\cdm)$.
We make the following observations:
\begin{itemize}
\item 
Let $z \in \RR$ be any location which is representative for the set of districts $\calD_z \neq \varnothing$. 
By property (P2), for any district $d \in \calD_z$, we have that $\delta(z,o) = \delta(x_i,z) + \delta(x_i,o)$ for every agent $i \in \calN_d$. 
Hence, by summing over all agents in the districts of $\calD_z$, we have 
\begin{align*}
\SC_z(z | \xx)  + \SC_z(o | \xx) = \delta(z,o) \cdot \lambda |\calD_z|.
\end{align*}
As the representative of each district $d \in \calD_z$, $z$ is the position of the median agent in $d$ and thus minimizes the total distance of the agents in $d$, that is, $\sum_{i \in \calN_d} \delta(x_i,z) \leq \sum_{i \in \calN_d} \delta(x_i,o)$. Hence, by summing over all districts in $\calD_z$, we have that
\begin{align*}
\SC_z(z | \xx) \leq \SC_z(o | \xx).
\end{align*} 
By combining the above two expressions, we obtain 
\begin{align}
\SC_z(z | \xx)  &=  \frac{1}{2} \delta(z,o) \cdot \lambda |\calD_z|. \label{eq:CDM-zz}
\end{align}
and 
\begin{align}
\SC_z(o | \xx) &\geq \frac{1}{2} \delta(z,o) \cdot \lambda |\calD_z|. \label{eq:CDM-zo}
\end{align}

\item   
Consider any alternative location $z \in Z \setminus\{w\}$. 
By property (P1), we have that $w$ is the left-most representative, and thus $z > w$.
By (P2), we have any agent $i$ in a district of $\calD_z$ lies in the interval defined by $z$ and $o$, which means that 
\begin{itemize}
\item $\delta(x_i,w) \leq \delta(w,o)$ if $z \leq o$, and 
\item $\delta(x_i,w) \leq \delta(w,z) = \delta(w,o) + \delta(z,o)$ if $z > o$. 
\end{itemize}
Since $\delta(z,o) \geq 0$, by summing over all the agents in the districts of $\calD_z$, we obtain that
\begin{align}\label{eq:CDM-zw}
\SC_z(w | \xx) &\leq \bigg( \delta(w,o) + \delta(z,o) \bigg) \cdot \lambda |\calD_z|.
\end{align}

\item
Since $w$ is the left-most representative (implied by (P1)) and the median among all representatives (which is why it is selected by the mechanism), it must be the case that $w$ is the representative of more than half of the districts in $Z$, and thus 
\begin{align}\label{eq:CDM-wz-size}
|\calD_w| \geq \sum_{z \in Z\setminus\{w\}} |\calD_z|.
\end{align} 
\end{itemize}

Given the above observations, we will now upper-bound the social cost of $w$ and lower-bound the social cost of $o$. 
By the definition of $\SC(w | \xx)$, and by applying \eqref{eq:CDM-zz} for $y=w$ and \eqref{eq:CDM-zw} for $z \neq w$, 
we obtain
\begin{align*}
\SC(w|\xx) &= \SC_w(w | \xx) + \sum_{z \in Z \setminus\{w\}} \SC_z(w|\xx) \\
&\leq \frac{1}{2} \delta(w,o) \cdot \lambda |\calD_w| 
+ \sum_{z \in Z \setminus\{w\}} \bigg( \delta(w,o) + \delta(z,o) \bigg) \cdot \lambda |\calD_z|  \\
&= \frac{1}{2} \delta(w,o) \cdot \lambda |\calD_w| 
+ \delta(w,o) \cdot \lambda \sum_{z \in Z \setminus\{w\}} |\calD_z| 
+  \sum_{z \in Z \setminus\{w\}}  \delta(z,o) \cdot \lambda  |\calD_z|. 
\end{align*}
By \eqref{eq:CDM-wz-size}, we further have that
\begin{align}
\SC(w | \xx) &\leq \frac{3}{2} \delta(w,o) \cdot \lambda |\calD_w| 
+  \sum_{z \in Z \setminus\{w\}}  \delta(z,o) \cdot \lambda  |\calD_z|  \nonumber \\
&\leq \frac{3}{2} \sum_{z \in Z}  \delta(z,o) \cdot \lambda  |\calD_z|.  \label{eq:CDM-mech}
\end{align}
On the other hand, by the definition of $\SC(o | \xx)$ and by applying \eqref{eq:CDM-zo}, we can lower-bound the optimal social cost as follows:
\begin{align}\label{eq:CDM-opt}
\SC(o | \xx) &= \sum_{z \in Z} \SC_z(o|\xx) \geq \frac{1}{2} \sum_{z \in Z}  \delta(z,o) \cdot \lambda  |\calD_z|.
\end{align}
Consequently, by combining \eqref{eq:CDM-mech} and \eqref{eq:CDM-opt}, the distortion of the instance $\calI$ subject to $\cdm$ is
\begin{align*}
\dist(\calI | \cdm) = \frac{\SC(w | \xx)}{\SC(o | \xx)} 
\leq 3.
\end{align*}
Since $\calI$ is an arbitrary instance of $\wc(\cdm)$, \cref{lem:continuous-structure} implies $\dist(\cdm) \leq 3$.
\end{proof}

\subsection*{Proof of \cref{lem:contsplemma}}
Since $\calM$ is strategyproof, by \cref{lem:spiout}, it is also strategyproof within districts. 
We enumerate the agents in $S$ as $\{1, ...,|S|\}$, and consider a sequence of district position profiles  $\{ \xx_d^{(1)}, \xx_d^{(2)}, ..., \xx_d^{(|S|)} \}$ such that $\xx_d^{(\ell)}$ is the same as $\xx$, with the exception that the first $\ell \in [|S|]$ agents of $S$ are now positioned at $y$. 
Hence, $\xx_d^{(|S|)} = \yy_d$. We will argue that for every district position profile $\xx_d^{(\ell)}$ in the sequence, the representative has to be $y$. 

First, consider $\xx_d^{(1)}$ and the corresponding agent $j \in \calN_d$ (agent $1$ in $S$), who is moved from $x_j$ to $y_j = y$. Suppose that $y$ is not the representative of $d$ under $\xx_d^{(1)}$. Then, if the true district position profile were $\xx^{(1)}$, agent $j$ would have incentive to misreport her position as being $x_j$ instead of $y_j=y$, so that the district position profile becomes $\xx$ and the representative of $d$ changes to her true position $y$, thus violating strategyproofness within districts. Using this, we can now easily show the statement by induction. In particular, assuming that $y$ is the representative of $d$ under district position profile $\xx_d^{(\ell-1)}$ for every $\ell \in [|S|]$, we can apply the same argument for the corresponding agent who is moved to obtain $\xx_d^{(\ell)}$.
\hfill $\qed$

\subsection*{Proof of \cref{thm:asymmetric}}
As we have already said, the structure of worst-case equilibria is the same as in the symmetric case. Therefore, the proof follows by the very same arguments used in the proofs of Theorems~\ref{thm:DMM-distortion}, \ref{thm:DDM-distortion} and \ref{thm:continuous-mechanism}. The main difference is that for every district $d \in \calD_z$, $\lambda$ (which is the size of every district in the symmetric case) will now be substituted by $n_d$. As a result, the inequalities will include the term 
$\sum_{d \in \calD_z} n_d$ instead of the term $\lambda |\calD_z|$. So, to obtain the desired bound on the distortion, we will use the fact that 
$$|\calD_z| \cdot \min_{d \in \calD_z} n_d \leq \sum_{d \in \calD_z} n_d \leq |\calD_z| \cdot \max_{d \in \calD_z}.$$ 

Let us now demonstrate exactly how the inequalities used in the proofs will change. We do this for Theorem~\ref{thm:DMM-distortion} which bounds the distortion of $\dmm$ in the discrete case; the inequalities used in the proofs of Theorems~\ref{thm:DMM-distortion} and \ref{thm:continuous-mechanism} change similarly. We have:
\begin{itemize}
\item Inequality \eqref{eq:DMM-zz} becomes
\begin{align}
\SC_z(z | \xx)  
&\leq  \frac{1}{2} \delta(z,o) \cdot \sum_{d \in \calD_z} n_d 
\leq  \frac{1}{2} \delta(z,o) \cdot |\calD_z| \cdot \max_{d \in \calD_z}. \label{eq:DMM-zz-asymmetric}
\end{align}

\item Inequality \eqref{eq:DMM-zo} becomes
\begin{align}
\SC_z(o | \xx) 
&\geq \frac{1}{2} \delta(z,o) \cdot \sum_{d \in \calD_z} n_d 
\geq \frac{1}{2} \delta(z,o) \cdot |\calD_z| \cdot \min_{d \in \calD_z} n_d. \label{eq:DMM-zo-asymmetric}
\end{align}

\item Inequality \eqref{eq:DMM-zw} becomes
\begin{align}
\SC_z(w | \xx) 
&\leq \bigg( \delta(w,o) + \delta(z,o) \bigg) \cdot \sum_{d \in \calD_z} n_d 
\leq \bigg( \delta(w,o) + \delta(z,o) \bigg) \cdot |\calD_z| \cdot \max_{d \in \calD_z} \label{eq:DMM-zw-asymmetric}
\end{align}

\item Inequality \eqref{eq:DMM-wz-size} remains the same:
\begin{align*}
|\calD_w| \geq \sum_{z \in Z\setminus\{w\}} |\calD_z|.
\end{align*} 
\end{itemize}
Hence, to lower bound $\SC(w | \xx)$, we now apply \eqref{eq:DMM-zz-asymmetric} for $z=w$, \eqref{eq:DMM-zw-asymmetric} for $z \neq w$, and \eqref{eq:DMM-wz-size}:
\begin{align}
\SC(w|\xx) 
&= \SC_w(w | \xx) + \sum_{z \in Z \setminus\{w\}} \SC_z(w|\xx) \nonumber \\
&\leq \delta(w,o) \bigg( \frac{1}{2} |\calD_w| + \! \! \sum_{z \in Z \setminus\{w\}} |\calD_z| \bigg) \cdot \max_{d \in 
\calD} n_d \nonumber 
 +  \bigg( \sum_{z \in Z \setminus\{w\}}  \delta(z,o) |\calD_z| \bigg) \cdot \max_{d \in \calD} n_d  \nonumber \\
&\leq \bigg( \frac{3}{2} \delta(w,o)  |\calD_w| 
+ \! \! \sum_{z \in Z \setminus\{w\}}  \delta(z,o)  |\calD_z| \bigg) \cdot \max_{d \in \calD} n_d  \nonumber   \nonumber \\
&\leq \frac{3}{2} \bigg( \sum_{z \in Z}  \delta(z,o)  |\calD_z| \bigg) \cdot \max_{d \in \calD} n_d.  \label{eq:DMM-mech-asymmetric}
\end{align}
To lower bound $\SC(o | \xx)$ we apply \eqref{eq:DMM-zo-asymmetric}:
\begin{align}
\SC(o | \xx) 
= \sum_{z \in Z} \SC_z(o|\xx)  
&\geq \frac{1}{2} \bigg( \sum_{z \in Z}  \delta(z,o) |\calD_z| \bigg) \cdot \min_{d \in \calD} n_d. \label{eq:DMM-opt-asymmetric}
\end{align}
So, by \eqref{eq:DMM-mech-asymmetric} and \eqref{eq:DMM-opt-asymmetric}, the distortion of $\calI$ subject to $\dmm$ is at most $3\alpha$, where $\alpha = \frac{\max_{d\in \calD}n_d}{\min_{d\in \calD}n_d}$. Since $\calI$ is an arbitrary instance of $\wc(\ddm)$, Lemma~\ref{lem:discrete-structure} implies the same upper bound on the distortion of the mechanism.
\hfill $\qed$

\end{document}